\DocumentMetadata{}
\documentclass{article} 
\AtBeginDocument{%
  }
\usepackage[T1]{fontenc}
\usepackage{hyperref}
\usepackage{url}
\usepackage{geometry} 
\usepackage{graphicx}
\usepackage{wrapfig}
\usepackage{threeparttable}
\usepackage{amsmath}
\usepackage{amsthm}
\usepackage{amsxtra}
\usepackage{pdfpages}
\usepackage{amsfonts} 
\usepackage{pdfpages}
\usepackage{bbm}
\usepackage{verbatim} 
\usepackage{microtype}
\usepackage{booktabs}       
\usepackage{comment}
\usepackage{xspace}
\usepackage{enumitem}
\usepackage{algorithm}
\usepackage[noend]{algpseudocode}
\usepackage{subcaption}
\usepackage{cleveref}

\usepackage{natbib}
\usepackage{authblk}

\algrenewcommand\algorithmicrequire{\textbf{Input:}}
\algrenewcommand\algorithmicensure{\textbf{Output:}}
\algrenewcommand\algorithmiccomment[1]{\hfill {\color{gray}\% #1}}

\DeclareMathOperator*{\argmin}{argmin}
\newtheorem{lemma}{Lemma}
\newtheorem{theorem}{Theorem}
\newtheorem{definition}{Definition}
\newtheorem{claim}[lemma]{Claim}

\DeclareMathOperator{\poly}{poly}

\newcommand{\NP}{\ensuremath{\mathbf{NP}}\xspace}

\newcommand{\wtwo}{\ensuremath{\mathbf{W[2]}}\xspace}

\newcommand{\labc}{\ensuremath{\ell_{{\cClustercenter}}}\xspace}
\newcommand\lab[1]{\ensuremath{\ell_{{#1}}}\xspace}
\newcommand\lh{\ensuremath{\ell_{{\histClustercenter}}}\xspace}
\newcommand\lopt{\ensuremath{\ell_{{\mathcal{C}^*}}}\xspace}
\newcommand\lalg{\ensuremath{\ell_{C}}\xspace}
\newcommand\lalga{\ensuremath{\ell_{C'}}\xspace}

\newcommand{\costp}[1]{\ensuremath{\mathrm{cost}_{#1}}\xspace}

\newcommand{\AlgCluster}{\ensuremath{\mathcal{C}}\xspace}
\newcommand{\AlgClustera}{\ensuremath{\mathcal{C}'}\xspace}
\newcommand{\cCluster}{\ensuremath{\mathcal{C}}\xspace}
\newcommand{\cluster}[1]{\ensuremath{\mathcal{C}_#1}\xspace}
\newcommand{\histCluster}{\ensuremath{\mathcal{H}}\xspace}

\newcommand{\optCluster}{\ensuremath{\cCluster^*}\xspace}

\newcommand{\AlgClustercenter}{\ensuremath{C}\xspace}
\newcommand{\AlgClustercentera}{\ensuremath{C'}\xspace}
\newcommand{\cClustercenter}{\ensuremath{C}\xspace}

\newcommand{\clustercenter}[1]{\ensuremath{C_{#1}}\xspace}
\newcommand{\histClustercenter}{\ensuremath{H}\xspace}

\newcommand{\optClustercenter}{\ensuremath{\cClustercenter^*}\xspace}

\newcommand{\points}{\ensuremath{{X}}\xspace}
\newcommand{\newPoints}{\ensuremath{{X}'}\xspace}
\newcommand{\pointsY}{\ensuremath{{Y}}\xspace}
\newcommand{\point}{\ensuremath{x}\xspace}
\newcommand{\pointy}{\ensuremath{y}\xspace}
\newcommand{\nbPoints}{\ensuremath{n}\xspace}

\newcommand{\dist}{\ensuremath{d}\xspace}
\newcommand{\clusterDist}{\ensuremath{\Delta}\xspace}

\newcommand{\nbUpdates}{\ensuremath{b}\xspace}
\newcommand{\nbClusters}{\ensuremath{k}\xspace}

\newcommand{\optRadius}{\ensuremath{r^*}\xspace}

\newcommand{\ball}{\text{{Ball}}\xspace}
\newcommand{\cI}{\mathcal{I}}
\newcommand{\cJ}{\mathcal{J}}

\newcommand{\mainAlgo}{\textsc{Greedy\-And\-Project}\xspace}

\newcommand{\ourproblem}{\textsc{Label-consistent $k$-center}\xspace}
\newcommand{\labcons}{\textsc{Label-consistent $k$-clustering}\xspace}
\newcommand{\kcenter}{\ensuremath{k}-{\allowbreak}\textsc{C}\textsc{enter}\xspace}
\newcommand{\gonz}{\ensuremath{\text{\textsc{\large fft}}}\xspace}
\newcommand{\carv}{\ensuremath{\text{\textsc{CARVE}}}\xspace}
\newcommand{\twotwoapprox}{\textsc{Over}\-\textsc{Cover}\xspace}
\newcommand{\resilient}{\ensuremath{\text{R\textsc{\large e\-sil\-ient}}}\xspace}
\newcommand{\consinst}{\ensuremath{((X,d),k,\histCluster=(\histClustercenter,\lh),b)}}

\newcommand{\mpara}[1]{\medskip\noindent\textbf{#1}}
\newcommand{\spara}[1]{\smallskip\noindent\textbf{#1}}
\newcommand{\para}[1]{\noindent\textbf{#1}}

\begin{document}

\title{Label-consistent clustering for evolving data}

\author[1]{Ameet Gadekar%
}
\author[2,3]{Aristides Gionis%
}
\author[2,3]{Thibault Marette%
}

\affil[1]{CISPA Helmholtz Center for Information Security, Saarbr\"{u}cken, Germany}
\affil[2]{KTH Royal Institute of Technology, Stockholm, Sweden}
\affil[3]{Digital Futures, Stockholm, Sweden}
\affil[ ]{\texttt{ameet.gadekar@cispa.de},  \texttt{argioni@kth.se},~\texttt{marette@kth.se}}
\date{}

\maketitle

\begin{abstract}
Data analysis often involves an iterative process, 
where solutions must be continuously refined in response to new data. 
Typically, as new data becomes available, an existing solution 
must be updated to incorporate the latest information. 
In addition to seeking a high-quality solution for the task at hand, 
it is also crucial to ensure consistency by minimizing drastic changes from previous solutions. 
Applying this approach across many iterations, ensures that the solution evolves gradually and smoothly.

In this paper, we study the above problem in the context of clustering, 
specifically focusing on the $k$-center problem. 
More precisely, we study the following problem: 
Given a set of points \points, parameters \nbClusters and \nbUpdates, and a prior
clustering solution \histCluster for~\points, 
our goal is to compute a new solution \cCluster for \points, consisting of \nbClusters centers, 
which minimizes the clustering cost while introducing at most \nbUpdates changes from \histCluster. 

We refer to this problem as \emph{label-consistent $k$-center}, 
and we propose two constant-factor approximation algorithms for it.
We complement our theoretical findings with an experimental evaluation 
demonstrating the effectiveness of our methods on real-world~datasets.
\end{abstract}

\section{Introduction}

Data clustering is a fundamental problem in data analysis with numerous applications across various 
domains~\citep{gan2020data}.
Traditionally, clustering is studied in a {one-shot} setting, 
where the goal is to find an optimal solution for a given dataset without incorporating prior information. 
However, in many applications, clustering is often an iterative process.
Instead of starting from scratch, a clustering solution from a previous computation may be available. 
The objective then becomes to \emph{refine} that solution in order to be \emph{consistent} with newly-introduced or modified~data. 
In such a scenario, 
we are interested in finding a high-quality solution,
which is also \emph{close enough} to the previously-available solution, 
so as to preserve stability and continuity.

As an example, consider clustering a stream of evolving data, such as news articles.
To effectively monitor major news stories, we want to update our clustering daily as new articles arrive. 
However, rather than forming entirely new clusters each day, 
we seek to maintain continuity with the previous day's clustering. 
This approach is crucial because major news stories often unfold over multiple days.
By ensuring consistency in clustering, we can track evolving narratives more effectively, 
preserving the coherence of long-running~stories.

In this paper, we study the problem of refining an existing (but possibly obsolete) clustering 
in the presence of new data, 
while respecting a consistency requirement between the refined solution and the existing~one.

More specifically, we introduce a novel problem, 
which we call \emph{label-consistent $k$-clustering}, defined as follows.
Given a set of points~\points, parameters \nbClusters and~\nbUpdates, and a prior 
clustering solution \histCluster for~\points, 
our goal is to compute a new solution \cCluster for \points, consisting of \nbClusters centers
and minimizing the clustering cost while introducing at most \nbUpdates changes to \histCluster. 
The type of changes that we account for is \emph{labeling changes}, 
namely, data points that are re-assigned to a different cluster;
see~\Cref{def:labcons} for a formal definition.
We focus on consistency in the context of the classic $k$-center problem \cite{gonzalez1985clustering}, to define \emph{label-consistent $k$-center}.

We refer to the prior solution \histCluster as \emph{historical clustering}. 
Our formulation is agnostic to how the historical clustering \histCluster has been computed, 
we only assume that the data points in \points can be assigned to centers in \histCluster.
The problem we study is motivated by scenarios where the historical clustering is \emph{fixed and immutable}, that is, it has already been presented to the user (e.g., a set of clients assigned to services, which has already been published) and cannot be retracted. In such a case, even if the historical clustering is imperfect, in order to avoid disruptions we want to require that the new clustering stays consistent to the historical clustering.

In the case of clustering evolving data, our framework can be instantiated as follows.
At time $t$ a clustering solution $\cCluster_{t}$ is computed recursively on data $\points_{t}$, using solution $\cCluster_{t-1}$ as historical clustering, 
where at time~$0$ any standard clustering method can be used to compute $\cCluster_{0}$. At time $t$, we want to refine our clustering solution for data $\points_{t}$, which may have additions and/or deletions with respect to $\points_{t-1}$. For clarity, first consider the case when $X_t$  contains only new points, i.e., no deletions. In that case, we extend the solution $\cCluster_{t-1}$ by assigning the newly added points to their closest centers.
Now if $X_t$ deletes some points from $X_{t-1}$, we can simply remove them from the cluster assignment in solution $\cCluster_{t-1}$, provided they were not centers in $\cCluster_{t-1}$. Otherwise, for each deleted point $p$ that was a center in $\cCluster_{t-1}$, we select another point from $\cCluster_{t-1}$ (that is still present in $X_t$) as a new cluster center. This ensures that all remaining points in the cluster centered at $p$ in $\cCluster_{t-1}$
are still present together in a cluster with a new~center.\footnote{Note that if all the points of the cluster centered at $p$ in $\cCluster_{t-1}$ are deleted in $X_t$, then we can simply remove the entire cluster from $\cCluster_{t-1}$.}

We can now apply our \emph{label-consistent $k$-clustering} problem 
and seek a new solution $\cCluster_{t}$ that minimizes the clustering cost 
while relabeling at most $b$ data points in $\points_{t}$
with respect to their current assignment to centers in \histCluster. This captures the essence of \histCluster, while allowing us to derive insights from the new data.

Prior work has considered similar frameworks for ensuring stability between 
clustering solutions, including 
a different formulation of consistent $k$-clustering~\citep{lattanzi2017consistent}, 
notions of evolutionary clustering~\citep{chakrabarti2006evolutionary}, 
and a recent work on resilient $k$-clustering~\citep{ahmadian2024resilient}.
More discussion on the differences of those approaches from our formulation 
is provided in the related-work section, 
while an empirical evaluation with the resilient $k$-clustering approach
is presented in our experiments.

Given that the problem we study is \NP-hard, we present two constant-factor approximation algorithms for \emph{label-consistent $k$-center}. 
The first algorithm, named \twotwoapprox, is a tight $2$-approximation algorithm, albeit running in $2^k\, poly(n)$.
Our second algorithm, named \mainAlgo, is a $3$-approximation
algorithm that runs in polynomial time and serves as the main result of the paper.
In summary, in this paper we make the following contributions.
\begin{itemize} 
\setlength\itemsep{-0.2em}
\item 
We introduce the problem of \emph{label-consistent $k$-clustering}, 
a novel clustering formulation that aims to optimize clustering cost 
while ensuring a consistency constraint, in the form of maximum number of data point re-labelings, from a historical clustering.
\item We present two constant-factor approximation algorithms for the $k$-center variant of the proposed problem (\Cref{theorem:2:2} and \Cref{thm:mainthm}).

\item 
We present a thorough experimental evaluation on real-world datasets and on different settings, 
comparing our algorithms with standard $k$-center algorithms
and a state-of-the-art baseline.
\end{itemize}

\section{Related work}
\label{section:related-work}

Data clustering is a widely-studied topic
\citep{gan2020data}.
Most related to our paper are methods with provable approximation guarantees. 
In this area, research has focused on problems, such as
$k$-means~\citep{ahmadian2019better}, 
$k$-median~\citep{charikar1999constant}, and 
$k$-center~\citep{gonzalez1985clustering}.
For the $k$-center problem, which is the focus of this paper, 
the classic \carv algorithm~\citep{hochbaum1985best}
provides a 2-approximation. 
This algorithm is a recurrent subroutine for the methods we develop in this paper.

Traditional methods do not provide
consistency requirements, and thus, a small change in the data
might result in a large discrepancy on the clustering result. 
Since consistency is a desirable property in many applications, 
researchers have also considered formulations to account for this requirement.

\mpara{Evolutionary clustering.} 
Different definitions have been proposed to measure temporal smoothness for clustering evolving data.
\citet{chakrabarti2006evolutionary} introduce a general framework
where they 
impose constraints on the changes of clustering solutions between consecutive snapshots. 
\citet{chi2009evolutionary} present a spectral approach, 
where the temporal-smoothness component of their approach accounts for cluster membership.
These approaches, as well as followup research~\citep{folino2013evolutionary,xu2014adaptive}, 
present mainly heuristic methods that do not offer quality guarantees.

\mpara{Dynamic clustering.} 
There is extensive work in the clustering literature proposing different formulations and methods to minimize the number of updates (recourse) for evolving data~\citep{chan2018fully,pellizzoni2023fully,bateni2023optimal}. However, existing approaches differ fundamentally from ours. These dynamic algorithms treat label-consistency as a soft objective to be minimized, either per update or across updates for amortized cost, rather than as a hard constraint enforced in the problem definition.
Furthermore, our problem is motivated by a regime where the number of points is extremely large compared to the number of time steps in which the input changes. This stands in contrast to typical dynamic settings, which often assume many updates and therefore require sublinear or near-sublinear running times.

\mpara{Low-recourse algorithms.} More generally, low-recourse algorithms have been proposed for problems other than clustering, and share similar goals with our approach. For instance,\cite{bhattacharya2023chasing} consider online bipartite matching with constraints on the amortized number of replacements, achieving logarithmic recourse bounds. Similarly, \cite{bernstein2019online} investigate the problem of maintaining solutions to time-varying optimization problems while minimizing the cost of transitions between solutions. Similarly to our difference with dynamic clustering, this line of work typically treats recourse as a soft objective to be minimized or bounds it in an amortized sense, whereas recourse is a hard constraint in our problem definition.

\mpara{Consistent clustering (a different definition).}
\citet{lattanzi2017consistent} introduce the idea of \emph{consistent clustering}, 
but their notion of consistency is significantly different than ours.
In particular, they measure consistency in terms of the symmetric differences between the sets of cluster centers, 
while we measure changes in the assignment of data points to centers.
We elaborate on the difference with our definition in Section \ref{section:problem}.
Furthermore, \citet{lattanzi2017consistent} and followup 
research~\citep{fichtenberger2021consistent,lkacki2024fully,forster2025dynamic}
consider an online setting and provide competitive-ratio analysis.

\mpara{Resilient clustering.} 
More recently, \citet{ahmadian2024resilient} introduced the notion of \emph{resilient $k$-clustering}.
Their goal is to find a clustering that is resilient to perturbations in the input data.
Their resiliency definition asks to preserve cluster membership of the data points, 
similarly to our consistency definition. 
The main difference is that they aim to achieve resiliency \emph{without any prior knowledge to historical clustering}; 
this makes their setting more strict,
and as we will see in our experimental evaluation, 
their algorithm often gives poor results in terms of resiliency/consistency measure.
This line of research falls in the topic of 
\emph{perturbation resilience}~\citep{balcan2016clustering,chekuri2018perturbation,bandyapadhyay2022perturbation}, 
where the goal is to find an optimal clustering that does not change
when the data are perturbed by a small amount. 
Instead, we do not make any assumption about the input data.

\section{\labcons}
\label{section:problem}

In this section, we introduce our notation  
and present the formal definition of the clustering problem we~study.

\mpara{Clustering.} 
Let $(\points,\dist)$ be a metric space on $n$ points.
For any $\pointsY\subseteq\points$ and $\point\in\points$, 
we use $\dist(\point,\pointsY)$ to denote the minimum distance from \point to any point in \pointsY, 
that is, $\dist(\point,\pointsY)=\min_{\pointy\in\pointsY}\dist(\point,\pointy)$. 
A \nbClusters-clustering of a set of points $\points$ is a pair $\cCluster=(\cClustercenter,\labc)$, 
where $\cClustercenter \subseteq \points$ is a set of $\nbClusters$ \emph{cluster centers}, and 
$\labc: \points \mapsto \cClustercenter$ is \emph{labeling function}, 
which assigns each point $\point\in\points$ to a center $\labc(\point)\in\cClustercenter$.

\begin{definition}[The $k$-clustering problem]\label{def:kclus}
Fix some positive real $p \in \mathbb{R}_{\ge 1}$.
Given a metric space $(X,d)$ and a positive integer $k \in \mathbb{Z}_+$, the \nbClusters-Clustering problem
asks to find a \nbClusters-clustering $\cCluster=(\cClustercenter,\labc)$ of $X$ 
that minimizes the cost function 
$\costp{p}(\cCluster)=\left(\sum_{\point\in\points}\dist(\point,\labc(\point))^p\right)^{1/p}$.   
\end{definition}
When $p=1$, the objective yields the \emph{\nbClusters-median} problem, and
when $p=2$, it captures the \emph{\nbClusters-means} objective.
When $p\rightarrow \infty$, it recovers the \kcenter problem,
where 
$\costp{\infty}(\cCluster)=\max_{\point\in\points}\dist(\point,\labc(\point))$.

\mpara{Label-consistent clustering.} 
To present the proposed notion of label-consistent clustering, 
we start by defining the distance between two clusterings.

\begin{definition}[Distance between two clusterings] 
\label{definition:clustering-distance}
Let $\cluster{1} = (\clustercenter{1}, \lab{1})$ and $\cluster{2} = (\clustercenter{2}, \lab{2})$ 
be two $\nbClusters$-clusterings over the same set of points $\points$. 
We define the distance between 
$\cluster{1}$ and $\cluster{2}$ as 
$
\clusterDist(\cluster1,\cluster2) = |\{\point \in \points, \lab1(\point) \neq \lab2(\point)\}|,
$
namely, the number of points in \points that are assigned to different 
cluster centers in \cluster{1} and \cluster{2}. 
\end{definition}

In this work, we measure the distance between two clusterings $\cluster{1}$ and $\cluster{2}$
using the distance function $\clusterDist(\cluster1,\cluster2)$ of Definition~\ref{definition:clustering-distance}.
An alternative definition is 
the size of the symmetric difference $|\clustercenter{1} \triangle\,  \clustercenter{2}|$, 
which was used to define consistent clustering in online 
settings~\citep{lattanzi2017consistent}.
One can observe that 
the two clustering distance functions
$\clusterDist(\cluster1,\cluster2)$ and $|\clustercenter{1} \triangle\,  \clustercenter{2}|$
can take very different values.
Assume now  that a clustering solution \histCluster, 
which might have been computed over old data, is provided for the current data points \points. 
We want to \emph{refine} the historical clustering \histCluster
while enforcing a \emph{consistency constraint}.

To ensure consistency with an available historical clustering $\histCluster$, we use a parameter \nbUpdates to control the degree of label consistency in the resulting clustering,
and we require to find a \nbClusters-clustering \cCluster that has minimal clustering cost
while satisfying the consistency constraint $\clusterDist(\histCluster,\cCluster)\leq \nbUpdates$.
More formally, we define the following problem.
\begin{definition}[\labcons]\label{def:labcons} 
Fix some  $p \in \mathbb{R}_{\ge 1}$.
Given a metric space $(X,d)$, two positive integers $k,b \in \mathbb{Z}_+$, and a historical clustering $\histCluster$ of $\points$ with $k$ centers, the \labcons problem seeks to find a \nbClusters-clustering $\cCluster=(\cClustercenter,\labc)$ of \points that minimizes the cost function 
$\costp{p}(\cCluster)=\left(\sum_{\point\in\points}\dist(\point,\labc(\point))^p\right)^{1/p}$,   
while ensuring $\clusterDist(\histCluster,\cCluster)\leq \nbUpdates$.
\end{definition}
In this paper we focus on the $\costp{\infty}(\cdot)$ objective, 
namely the \emph{\nbClusters-center} objective.
\begin{definition}[\ourproblem]
The \ourproblem problem is the instantiation of the \labcons problem with $p=\infty$, 
and thus, $\costp{\infty}(\cCluster)=\max_{\point \in \points} \dist(\point,\labc(\point))$.
\end{definition}

We denote by $\cI=\consinst$  an instance of \ourproblem. We say a (clustering) solution $\cCluster=(\cClustercenter,\labc)$  to $\cI$ is \emph{feasible} if $|C|=k$ and $\clusterDist(\histCluster,\cCluster)\leq \nbUpdates$, i.e., $\cCluster$ opens $k$ centers and reassigns at most $b$ points from the historical clustering $\histCluster$. 
We denote by $\optCluster=(\optCluster,\lopt)$ as an optimal solution to $\cI$, and denote by $r^*:=\costp{\infty}(\optCluster)$, the optimal \emph{radius} of $\optCluster$.

An interesting observation is that an optimal solution \optCluster for \ourproblem
may not necessarily assign each point in~\points to its closest center in \optCluster, 
as the consistency constraint may force a data point to be assigned to its historical cluster center
although it is not the closest center.

It is easy to see (by setting $\nbUpdates=\nbPoints$) that, for specific values of $p$, \labcons inherits the hardness of the corresponding $k$-clustering problem. 
In fact, for \ourproblem, we obtain the following stronger hardness results,\footnote{Similar hardness of approximation results follow for \textsc{Label-Consistent} $k$-median and $k$-means.} whose proof is deferred to~\Cref{ss:hardproof}.

\begin{theorem}\label{thm:hardness} 
For every $\epsilon>0$ and every computable function $f$, the following hold:
\begin{enumerate}
    \setlength\itemsep{-0.2em}
    \item[($i$)] It is \NP-hard to approximate \ourproblem\ within a factor of $(2-\epsilon)$.
    \item[($ii$)] It is \wtwo-hard to approximate \ourproblem\ within a factor of $(2-\epsilon)$ 
    when parameterized by $k$~\citep{downey2012parameterized}.
    Thus, there is no $(2-\epsilon)$-approximation algorithm running in time $f(k)\,n^{O(1)}$. 
    \item[($iii$)] The above hardness results remain valid even if one is allowed to reassign $g(b)\ge b$ points, for any function $g$.
\end{enumerate}
\end{theorem}

\section{Algorithms}
\label{section:algorithms}
In this section, we present two constant-factor approximation algorithms for the \ourproblem problem. 
In~\Cref{ss:fptalg}, we present \twotwoapprox, a tight $2$-approximation algorithm that runs 
in FPT time in parameter $k$~\citep{downey2012parameterized}.
\mainAlgo, the main contribution of this paper, is presented in~\Cref{ss:mainalg}, 
and is a $3$-approxi\-mation algorithm that runs in polynomial time.
We remark that both of our algorithms work under the assumption that the optimal radius $\optRadius$ is known. In~\Cref{ss:assump} we discuss how this assumption can be removed and its relation to the budget parameter~$b$.

\mpara{Notation.}
Let $(X,d)$ be a metric space. For a point $x\in X$ and $r \in \mathbb{R}_{\ge 0}$, we denote by $\ball(x,r)$ as the set of points that are at a distance at most $r$ from $x$. $\poly(n)$ denotes a fixed function that is polynomial in $n$. 

\subsection{A tight $2$-approximation in FPT time}
\label{ss:fptalg}

In this section, we design an FPT $2$-approximation algorithm for \ourproblem. 

\begin{theorem}
\label{theorem:2:2}
There is a $2$-approximation algorithm for the \ourproblem problem running in time~$2^k \, \poly(n)$.
\end{theorem}

\para{Overview of the \twotwoapprox algorithm.}
The high level idea of our algorithm, is simple and clean: first, we assume that we know the optimal radius $r^*$ and the historical centers $H^*$ that are present in an optimal solution; next, we identify the points that are far from $H^*$ using $r^*$ and cluster them using any classical $2$-approximation \kcenter algorithm; finally, we reassign points to minimize the number of reassignments. A crucial part of the analysis is to show that \twotwoapprox opens at most $k$ centers.

In more detail, let $\optCluster=(\optClustercenter,\lopt)$ 
be an optimal solution to $\cI$ with cost $r^*$, and 
let $H^* = \optClustercenter \cap \histClustercenter$, 
be the set of historical centers present in $\optClustercenter$.
First, the algorithm guesses $H^*$.
Next, it considers all the points in $X_B:=X \setminus \cup_{h \in H^*}\ball(h,r^*)$, and 
clusters them using \carv algorithm of~\cite{hochbaum1985best}. 
The \carv algorithm,
when given a dataset $X'$ and radius $r$, 
repeatedly picks a center $s'\in X'$ and deletes all the points within distance $2r$ from $s'$, 
until all the points are deleted. 
If $S'$ is the set of picked centers by \carv, then it is easy to see that $d(x',S') \le 2r$ 
for every $x' \in X'$. 
Finally,
the sets $S'$ and $H^*$ are merged
to obtain a center set $C$, and 
points are assigned to $C$ minimizing the number of reassignments. 
Pseudocode for \twotwoapprox and \carv are available in Appendix~\ref{appendix:peusodcode}, and the analysis is deferred to \Cref{ss:fptproof}.

\subsection{A $3$-approximation algorithm}
\label{ss:mainalg}

In this section, we present a polynomial time $3$-approximation algorithm for \ourproblem.

\begin{theorem}\label{thm:mainthm}
    There is a $3$-approximation algorithm for the \ourproblem problem 
    running in time $O(n^2 \log n + nk \log n)$.
\end{theorem}

The presentation proceeds as follows: 
we outline our algorithm \mainAlgo in~\Cref{ss:overview} and,
prove its correctness in~\Cref{ss:analysis}. The proof of~\Cref{thm:mainthm} is present in~\Cref{ss:mainthmproof}.

\begin{minipage}[t]{0.50\textwidth}
\begin{algorithm}[H]
\caption{\textsc{GreedyAndProject}}\label{alg:3approxnew}
\begin{algorithmic}[1]
\Require an instance $\cI=\consinst$ of \ourproblem,  optimal radius \optRadius
\Ensure a $3$-approximate solution $\AlgCluster=(\AlgClustercenter, \lalg)$ to $\cI$

\State $C_0,C_1,C \gets \emptyset$
\State $S \gets \carv(\points,2\optRadius)$\label{alg:3approxnew:fft} \Comment{\Cref{alg:fft}}
\For{$s \in S$\label{3approxnew:for}}
    
    \If{$N_\histClustercenter(s) \neq \emptyset$} 
    \State add the maximum weight center, $\tilde{s}$, from $N_\histClustercenter(s)$ to $C_0$ by breaking ties arbitrarily
    \Else\ add $s$ to $C_0$
    \EndIf
\EndFor
\State Let $C_1$ be the $(k-|C_0|)$ maximum weight historical centers from $\histClustercenter\setminus C_0$ \label{algo:3approx:c1}
\State Let $C \gets C_0 \cup C_1$ 
\For{$\point \in \points$\label{algo:3approx:assgn}} \Comment{assign points to $C$}
    \If{$\lh(\point) \in C$  and $d(x,\lh(x))\le \optRadius$}
         $\lalg(\point) = \lh(\point)  $
    \Else     
        \ $\lalg(\point) \gets \arg\min_{c\in C}d(x,c)$
    \EndIf    
\EndFor
\State \textbf{return} $\AlgCluster=(\AlgClustercenter, \lalg)$\label{algo:3approx:ret}
\end{algorithmic}
\end{algorithm}
\end{minipage}
\hfill
\begin{minipage}[t]{0.46\textwidth}
\begin{algorithm}[H]
\caption{\text{Analyzing~\Cref{alg:3approxnew}}}\label{alg:helperalgo}
\begin{algorithmic}[1]
\Require $\cI=\consinst$, optimal solution $\optCluster=(\optClustercenter,\lopt)$, set $S$ obtained from~\Cref{alg:3approxnew:fft} of~\Cref{alg:3approxnew}
\Ensure a feasible solution $\AlgClustera=(\AlgClustercentera, \lalga)$ to $\cI$
\State $C'_0,C'_1,C'_2, C' \gets \emptyset$
\For{$s \in S$\label{alg:helperalgo:for}\label{alg:helper:C0for}}
        \If{$N_\histClustercenter(s) \neq \emptyset$} 
        \State add the maximum weight center, $\hat{s}$, from $N_\histClustercenter(s)$ to $C'_0$ by breaking ties arbitrarily
    \Else\ add $s$ to $C'_0$    
    \EndIf
\EndFor
\For{$s\in S_\gamma$}  \Comment{$S_\gamma= \{s\in S \vert \Gamma^*_s \neq \emptyset\}$\label{alg:helper:forc1}}
    \State pick $|\Gamma^*_s|-1$ maximum-weight historical centers, $\hat{\Gamma}_s$, from $\Gamma^*_s \setminus \{\hat{s}\}$ \label{alg:helperalgo:pickadd}
    \State $C'_1 \gets C'_1 \cup \hat{\Gamma}_s$
\EndFor

\State Let $C'_2$ be the $|H^*_f|$  maximum weight historical centers from $\histClustercenter\setminus (C'_0 \cup C'_1)$ \label{alg:helperalgo:pickhid}
\State Let $C' \gets C'_0 \cup C'_1 \cup C'_2$ 
\For{$\point \in \points$\label{alg:helper:assgn}} 
    \If{$\lh(\point) \in C$  and $d(x,\lh(x))\le \optRadius$}
         $\lalg(\point) = \lh(\point)  $
    \Else     
        \ $\lalg(\point) \gets \arg\min_{c\in C}d(x,c)$
    \EndIf    
\EndFor
\State \textbf{return} $\AlgClustera=(\AlgClustercentera, \lalga)$
\end{algorithmic}
\end{algorithm}
\end{minipage}

\subsubsection{Overview of the \mainAlgo algorithm}
\label{ss:overview}

The pseudo-code of our \mainAlgo algorithm is described in~\Cref{alg:3approxnew}. 
We assume that the algorithm has access to the optimal radius $r^*$ of the input instance $\cI=\consinst$. Let the historical clusters be denoted as $\hat{\Pi}=\{\hat{\pi}_h\}_{h \in \histClustercenter}$, where $\hat{\pi}_h=\{x\in X, \lh(x) = h\}$. 
For a historical center $h \in \histClustercenter$, the \emph{weight} of $h$, denoted by $w(h)$, is the total number of points in $\ball(h,r^*) \cap \hat{\pi}_h$.
For $x \in X$, let $N_\histClustercenter(x)$ denote the set of historical centers within distance $r^*$ from $x$, i.e., $N_\histClustercenter(x) = \{h \in \histClustercenter \vert d(x,h) \le r^*\}$.
For a subset $T \subseteq \histClustercenter$ of historical centers,
let $w(T) = \sum_{h \in T} w(h)$. 
Let $\optCluster=(\optClustercenter,\lopt)$ be a fixed (but unknown) optimal solution to $\cI$.
The algorithm works in three~phases: 

\spara{($i$)~Greedy phase:} In this phase (\Cref{alg:3approxnew:fft}), the  algorithm computes a set $S$ of centers such that, for every point $x\in X$ it is $d(x,S) \le 2r^*$ (see \carv~\Cref{alg:fft}),  $|S| \le k$, and $d(s,s') > 2r^*$, for $s \neq s' \in S$.


\spara{($ii$)~Project phase: } In this phase (\textbf{for}-loop in~\Cref{3approxnew:for}), 
for every $s \in S$, the algorithm swaps $s$ with a maximum weight historical center in  $N_\histClustercenter(s)$, if it is non-empty.  Note that, if $s$ is served by a historical center in  $\optCluster$, then there exists a historical center in $N_\histClustercenter(s)$.

\spara{($iii$)~End phase: } In this phase (Lines~\ref{algo:3approx:c1}-\ref{algo:3approx:ret}), the algorithm first adds $(k-|S|)$ historical centers of maximum weight that have not yet been picked to its center set $C$ so that $|C|=k$. Then, it assigns every point $x \in X$ to $C$ minimizing the number of reassignments.

\medskip
It is easy to see that, for every $x \in X$, it holds that $d(x,C) \le 3r^*$. However, the crucial part is to show that the number of reassignments by the algorithm is at most $b$, i.e., $\clusterDist(\histCluster,\AlgCluster)\leq \nbUpdates$. The key observation for proving this claim is that in the Project phase,
if $s \in S$ is served by a historical center $h\in \optClustercenter$, then $h$ is a candidate for swapping $s$. Furthermore, since the algorithm picks the maximum-weight historical center~$\tilde{s}$, it holds that $w(\tilde{s}) \ge w(h)$. This inequality holds for every $s \in S$ as 
the set of centers that are within distance $r^*$ from $s \neq s' \in S$ are disjoint since $d(s,s') > 2r^*$. Finally, the algorithm adds $(k-|S|)$ historical centers of maximum weight from the remaining historical centers. Therefore, at every step, with respect to the number of reassignments,  the algorithm  does at least as good as the optimal solution.

\subsubsection{Analysis}\label{ss:analysis}

In this section, we prove the guarantees of the \mainAlgo algorithm (\Cref{alg:3approxnew}).
\begin{theorem}\label{thm:3approx}
Given an instance $\cI=\consinst$ of \ourproblem with optimal cost $r^*$, \Cref{alg:3approxnew} returns a feasible $3$-approximate solution   $\AlgCluster=(\AlgClustercenter, \lalg)$  to $\cI$ in time $O(nk)$.
\end{theorem}

\begin{proof}
First note that $|C| =k$, since we have that $|C_0| = |S| \le k$. Furthermore, for every $x \in X$, we have that $d(x,S) \le 2r^*$, and hence $d(x,C) \le d(x,S) + r^* = 3r^*$. Note that this concludes the proof, as either $\lh(x)\in C$ and $d(x,\lh(x))\leq \optRadius$, then $d(x,\labc(x))=d(x,\lh(x))\leq \optRadius$. Otherwise, $d(x,\labc(x))=d(x,C)\leq 3\optRadius$ as desired. Therefore, we only need to show that $\clusterDist(\AlgCluster,\histCluster)\le b$.

\mpara{Correctness:}
We start with some basic definitions. Recall that, for a historical center $h \in \histClustercenter$, the weight of $h$ is $w(h)= |\ball(h,r^*) \cap \hat{\pi}_h|$.
    Fix an optimal solution   $\optCluster=(\optClustercenter, \lopt)$ to $\cI$.
    We say $H^* := \optClustercenter \cap \histClustercenter$, the set of historical centers present in the optimal centers $\optClustercenter$, as the \emph{historical optimal centers}.
    Consider the set $S$ obtained from $ \carv(\points,2\optRadius)$, and $N_\histClustercenter(s) = \histClustercenter \cap \ball(s,\optRadius)$, for $s \in S$. 
    Then, note that, for $s\neq s' \in S$, we have that $N_\histClustercenter(s) \cap N_\histClustercenter(s') = \emptyset$, since $d(s,s')>2r^*$.
    Furthermore, we say that a historical optimal center $h \in H^*$ is \emph{covered} by $s \in S$, if $h \in N_\histClustercenter(s)$. Note that a historical optimal center can be covered by at most one point in $S$, but a point in $S$ can cover multiple historical optimal centers. Hence, for $s \in S$, let $\Gamma^*_s = H^* \cap N_\histClustercenter(s)$ be the set of historical optimal centers covered by $s$. 
    Next, let $H^*_c = \bigcup_{s\in S} \Gamma^*_s$ be the historical optimal centers covered by $S$, and let $H^*_f= H^* \setminus H^*_c$, which we call \emph{far} historical optimal centers to $S$.
    Finally, let $S_\gamma \subseteq S$, be the set of points of $S$ that covers at least one historical optimal center, i.e. for which $\Gamma^*_s \neq \emptyset$. Then, $\sum_{s \in S} |\Gamma^*_s| = \sum_{s \in S_\gamma} |\Gamma^*_s| = |H^*_c|$.

    
    To show that~\Cref{alg:3approxnew} returns a feasible solution, 
    we consider an \emph{idealized}~\Cref{alg:helperalgo} 
    that is powerful and knows the optimal solution $\optCluster$. 
    \Cref{alg:helperalgo} receives as input: 
    ($i$) the instance $\cI$, 
    ($ii$) the optimal solution~$\optCluster$, and 
    ($iii$) the set $S$ obtained from~\Cref{alg:3approxnew:fft} of~\Cref{alg:3approxnew}. 
    It computes three sets: $C'_0, C'_1$ and $C'_2$ and computes a set of centers 
    $C = \cup_{i \in [3]} C'_i$. 
    First, it sets $C'_0 =S$, and replaces $s \in S$ 
    with the maximum weight historical center $\hat{s} \in N_\histClustercenter(s)$  
    if $N_\histClustercenter(s) \neq \emptyset$ (\textbf{for}-loop in~\Cref{alg:helper:C0for}).
    Then, for every $s \in S_\gamma$, it adds $|\Gamma^*_s|-1$ maximum-weight historical centers from $\Gamma^*_s \setminus \{\hat{s}\}$ to $C'_1$ (\textbf{for} loop in~\Cref{alg:helper:forc1}). Finally, in~\Cref{alg:helperalgo:pickhid}, the algorithm picks $|H^*_f|$  maximum-weight historical centers from the remaining historical centers for $C'_2$. Note that given~$\optCluster$, the algorithm can compute $H^*, H^*_c,S_\gamma$, and $\Gamma^*_s$ for $s\in S_\gamma$.
    
    We first show that the solution of~\Cref{alg:helperalgo} is a feasible solution to $\cI$.

    \begin{lemma}\label{lem:helperlemma}
        The solution $\AlgClustera=(\AlgClustercentera, \lalga)$ returned by~\Cref{alg:helperalgo} is a feasible solution to $\cI$.
    \end{lemma}
    
Finally, we show that~\Cref{alg:3approxnew} reassigns no more points than \Cref{alg:helperalgo},
which is bounded by~$b$.

\begin{lemma}\label{lem:3apxvshelper}
    The number of points reassigned by solution $\AlgCluster$ of~\Cref{alg:3approxnew} is no more than the number of points reassigned by the solution $\AlgClustera$ of idealized~\Cref{alg:helperalgo}.
\end{lemma}

The proofs of~\Cref{lem:helperlemma} and~\Cref{lem:3apxvshelper} are deferred to~\Cref{ss:lmoneproof} and~\ref{ss:lemtwoproof}, respectively.

\emph{Running time:} The computation of set $S$ takes time $nk$, while the \textbf{for} loop takes time $k^2$.~\Cref{algo:3approx:c1} takes $O(k\log k)$ time, and the final assignment takes $O(nk)$ times. Therefore, the resulting time of~\Cref{alg:3approxnew} is $O(nk)$, finishing the proof of the theorem.
\end{proof}

\subsection{Guessing and verifying the optimal radius}\label{ss:assump}

\mpara{Guessing the optimal radius $\optRadius$.} As discussed in the previous sections, our algorithms assume that the optimal radius \optRadius is known and is given as input. This is a very common assumption in designing algorithms for \kcenter and its variants. This assumption can be removed by iterating over all~$n^2$ pairwise distances as a candidate for $\optRadius$, resulting in a multiplicative factor of $n^2$ in the overall running time. The running time can be sped up by sorting these distances and using binary search to find $\optRadius$. However, in practice, the distance aspect ratio $\Delta$, which is defined as the ratio of the maximum to the minimum distance, is often bounded polynomially in $n$. In such settings, it is standard to speed up the guessing of $\optRadius$ by discretizing all distances into powers of $(1+\epsilon)$, for a small $\epsilon>0$. This approach reduces the number of candidate radii from $n^2$ to $O(\log n/\epsilon)$, at the cost of introducing only a multiplicative $(1+\epsilon)$ factor in the approximation.

\mpara{Verifying the guess for the optimal radius $\optRadius$.} Since the idea is to run our algorithms (\Cref{alg:3approxnew} and \Cref{alg:twotwoapprox}) for every guess of $\optRadius$, we need to make sure that our algorithms return a \emph{feasible} solution.  
Note that our algorithms return the solution corresponding to the guess of \optRadius for which a feasible solution was found. To check the feasibility of a solution  $\AlgCluster=(\AlgClustercenter, \lalg)$ (\Cref{algo:3approx:ret} of~\Cref{alg:3approxnew} and \Cref{alg:fpt:ret} of~\Cref{alg:twotwoapprox}), the algorithm checks if $\clusterDist(\AlgCluster,\histCluster) \le b$.

\section{Experimental evaluation}
\label{section:experiments}

We empirically evaluate our algorithms, \twotwoapprox and \mainAlgo, on real-world datasets. We use four temporal datasets: \emph{Electric Consumption} \citep{individual_household_electric_power_consumption_235},  \emph{OnlineRetail} \citep{online_retail_352}, \emph{Twitter}  \citep{twitter_geospatial_data_1050}, and  \emph{Uber},%
\footnote{\url{https://www.kaggle.com/datasets/fivethirtyeight/uber-pickups-in-new-york-city}} as well as one non-temporal dataset, \emph{Abalone} \citep{abalone_1}. The description and preprocessing of these datasets is deferred to Appendix \ref{appendix:datasets}. 
Below we describe the experimental setup, 
baselines, 
and practical improvements of our methods. 
Finally, we present our results.

\subsection{Experimental setup}
\label{sec:setups}
We consider four  setups. In the main body of the paper, we focus on the first setup, which simulates accommodating historical clustering upon arrival of new data, and the second setup, which simulates temporal evolution of data. The third setup evaluates our algorithms on noisy data, and is presented in Appendix  \ref{appendix:eps_close_instances}, and the fourth setup explores scalability with respect to input parameters, as presented in Appendix \ref{appendix: scalability}.

\mpara{Baselines.} We evaluate our algorithms against three baselines:

\noindent
\text{(i) \carv} algorithm \cite{hochbaum1985best}, described in Algorithm \ref{alg:fft}.

\noindent
\text{(ii) \gonz} algorithm \cite{gonzalez1985clustering}, which starts with an arbitrary center, and repeats the following operations $k-1$ times: find the point furthest away from the current set of centers and add it to the set of centers.

\noindent
\text{(iii) Resilient $k$-clustering}. 
We implemented a version of the algorithm described by \citet{ahmadian2024resilient}. 
This algorithm first opens up to $\alpha \nbClusters$ centers in a resilient way, and then $\beta \nbClusters$ centers using \gonz, 
opening in total up to $(\alpha+\beta)\nbClusters$ centers.  
To ensure a fair comparison, we set $\alpha=\beta=0.5$, to open $k$ centers.
We  refer to this algorithm as~\resilient.

Next, we  detail the first two setups. We denote by \textsc{H{\large ist}} a historical clustering algorithm and  \textsc{A{\large lg}} an algorithm for \ourproblem. 

\spara{Setup 1: New data arrival.} 
Consider a dataset $\points$. First, we apply \gonz on $\points$ to find $3k$ clusters $X = C_1 \cup \ldots \cup C_{3k}$. Then, we consider the dataset $\newPoints = C_1 \cup \ldots \cup C_k\subseteq \points$. We run \textsc{H{\large ist}} on \newPoints and take the result to be the historical clustering~$\histCluster$, and define $\lh(x)$  for every $x\in\points\setminus\newPoints$ to be the closest center from $x$ in \histCluster. Finally, we run \textsc{A{\large lg}} on $X$ with \histCluster as historical clustering. We record the score of the output clustering, for fixed $k$ and varying $b$.

\spara{Setup 2: Evolutionary setting.} Consider a temporal dataset partitioned in $t$ slices $X=X_1\cup\ldots\cup X_t$, where every data slice has equal size. We sequentially cluster the data: first, we run \textsc{H{\large ist}} on $X_1$ to obtain $C_1$. Then, for every $i\in \{2,\ldots ,t\}$, we set $\histClustercenter=C_{i-1}$, and define $\lh(x)$ for every $x\in\points_i$ to be the closest center from $x$ in \histCluster, and we run \textsc{A{\large lg}} on $X_i$ with $\histCluster = (\histClustercenter, \lh)$ as historical clustering. In the experiments below, we set $t=20$. We record the score and the number of updates 
of the proposed solution on each of these data slices.

\mpara{Implementation choices.} In practice, we run both algorithms to get a solution with the aforementioned theoretical guarantees. Then, if there is budget remaining, we introduce a refinement phase, where we use the remaining budget to assign the furthest points to their closest center. Since this phase moves points to closer centers, it improves the objective value in practice, while preserving the theoretical guarantees.

Additionally, we modify \twotwoapprox to remove its exponential running time in $k$. Instead of considering every possible choice for $H^*$, we construct it greedily: we remove as many historical centers as the budget permits, and define the remaining historical centers as our guess of $H^*$. This modified version of  \twotwoapprox runs in polynomial time, though at the cost of losing its approximation guarantee.

\begin{figure*}[t]
\centering
\begin{subfigure}{\textwidth}
\includegraphics[width=.3\textwidth]{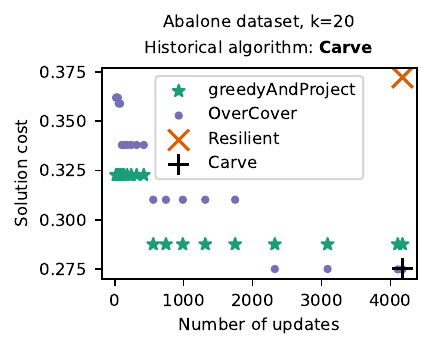}\hfill
\includegraphics[width=.3\textwidth]{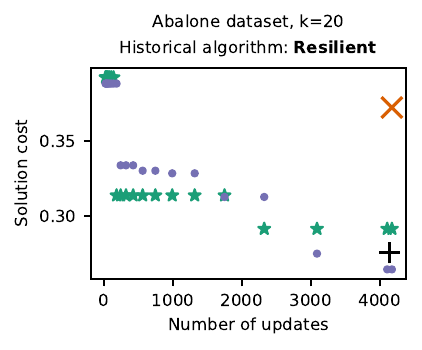}\hfill
\includegraphics[width=.3\textwidth]{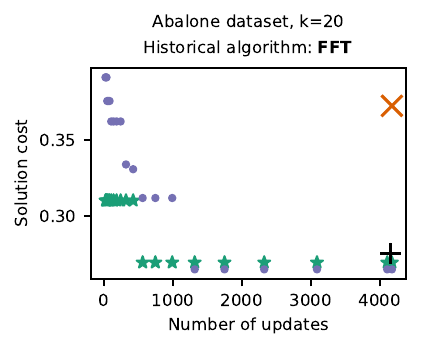}
\end{subfigure}
\begin{subfigure}{\textwidth}
\includegraphics[width=.3\textwidth]{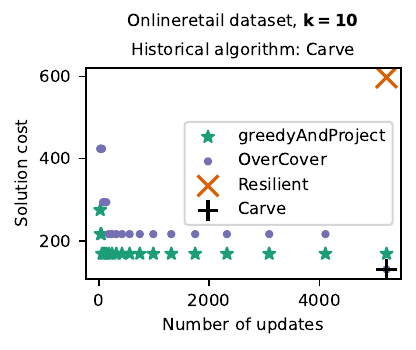}\hfill
\includegraphics[width=.3\textwidth]{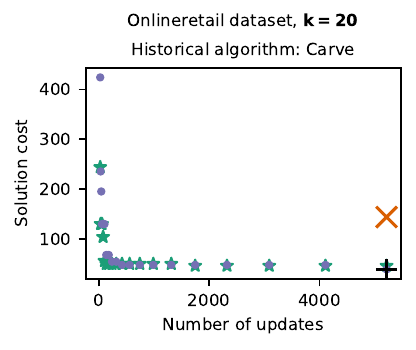}\hfill
\includegraphics[width=.3\textwidth]{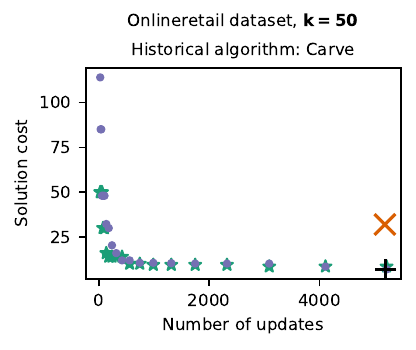}
\end{subfigure}
\hfill
\caption{Comparison of our algorithm with the baselines for the first experimental setup.}
\label{fig:setup1}
\end{figure*}

\begin{figure*}[t]
\centering
\begin{subfigure}{\textwidth}
\includegraphics[width=.58\textwidth]{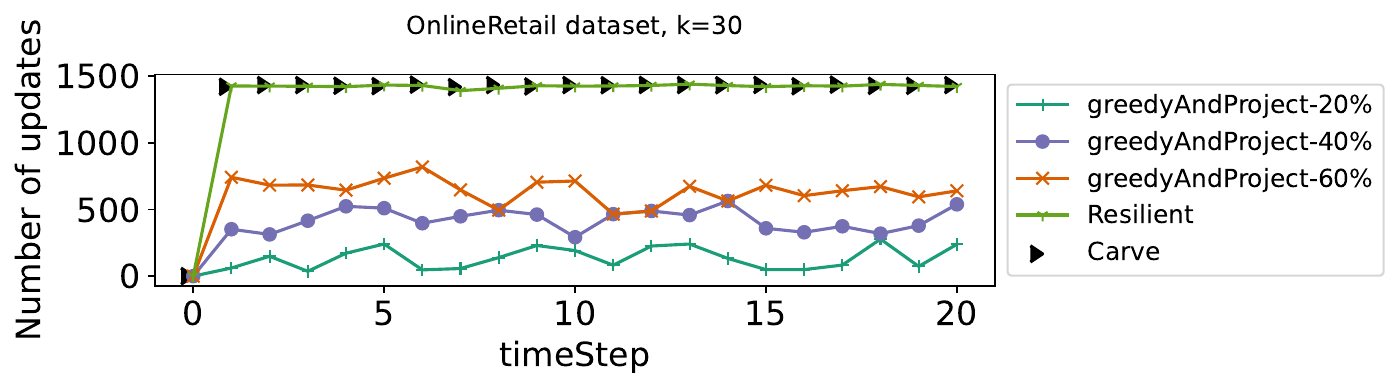}\hfill
\includegraphics[width=.41\textwidth]{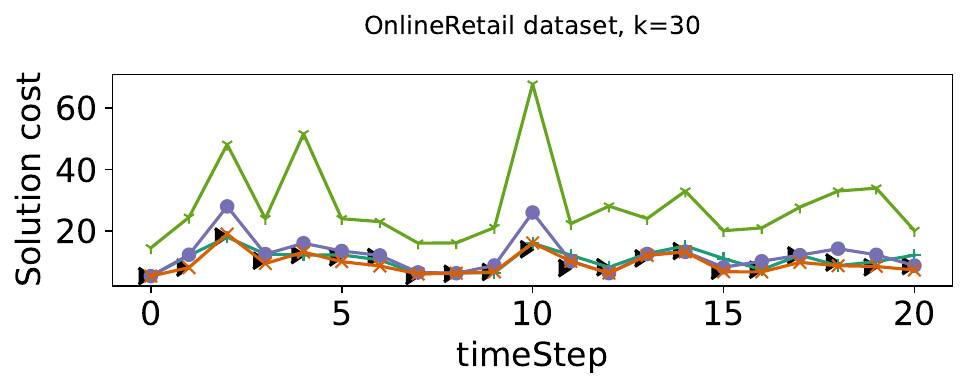}\hfill
\end{subfigure}
\hfill
\begin{subfigure}{\textwidth}
\includegraphics[width=.58\textwidth]{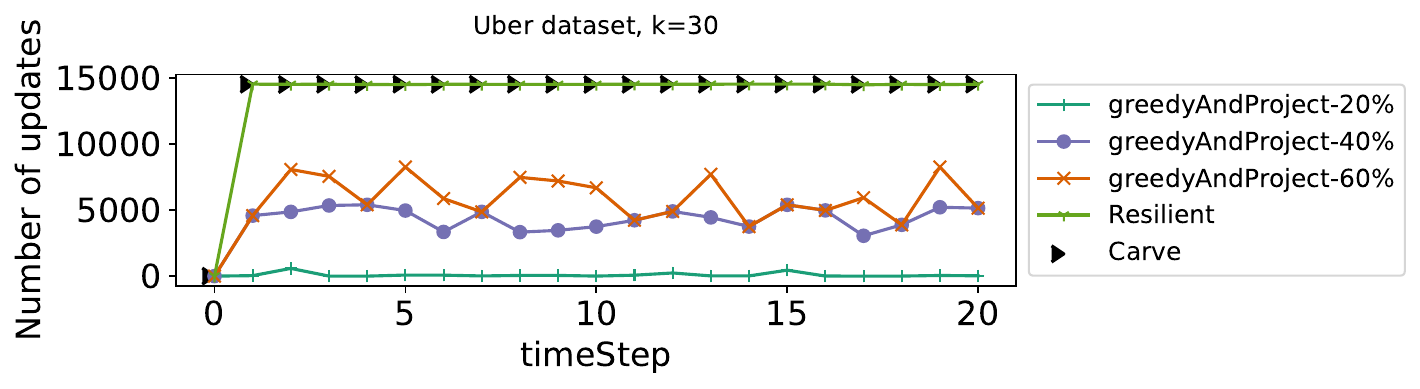}\hfill
\includegraphics[width=.41\textwidth]{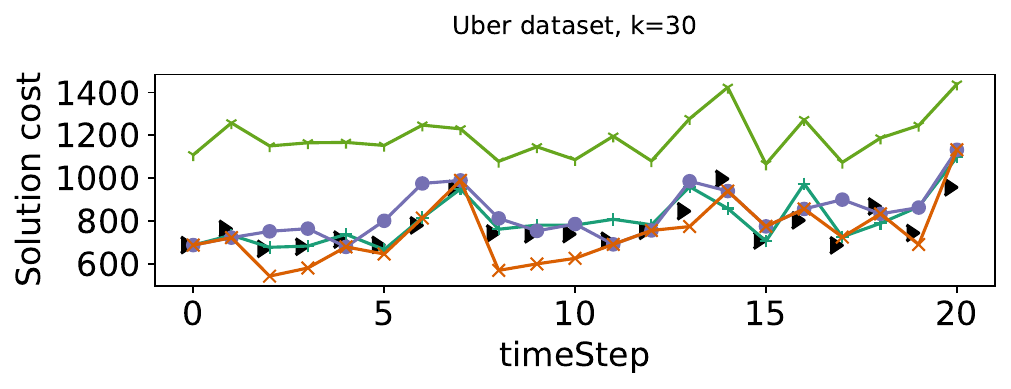}\hfill
\end{subfigure}
\caption{Comparison of our algorithm with the baselines for the second experimental setup.}
\label{fig:setup2}
\end{figure*}

\subsection{Empirical results}

In this section, we evaluate the output quality of our algorithms and baselines, for the two setups described above. We discuss results on Figure \ref{fig:setup1} and Figure \ref{fig:setup2}, while the complete results for the experiment is deferred to Appendix \ref{appendix:extraPlots}. To avoid clutter in the figures, we only report results of \resilient and \carv baselines, since \gonz performs similarly to \carv.

\spara{Setup 1.} In Figure \ref{fig:setup1}, we report results for the Abalone dataset, with varying historical clustering algorithm H\textsc{\large{ist}}$\in\{\carv, \gonz, \resilient\}$ and $k=20$, and results for the OnlineRetail dataset, with H\textsc{\large{ist}}$=$\carv and $k\in\{10,20,50\}$.

First, we evaluate the number of updates of each algorithms. Note that only one point per baseline algorithm is reported as their input ignores the budget $b$. We notice that almost all points are reassigned for all the baselines, while our algorithms respect the budget.

Next, in term of solution cost, both \mainAlgo and \twotwoapprox outperform \resilient, even for small number of updates, and perform similarly to \carv for larger values. For instance, for the OnlineRetail dataset, for $k=50$, 
reassigning as little as $9\%$ of the points brings the cost to within $44\%$ of that of \carv, while reassigning $49\%$ of the points further narrows the gap to just $22\%$.

Finally, \twotwoapprox usually achieves better cost than \mainAlgo for very high number of updates, since our implementation of \twotwoapprox coincides with \carv when $b=n$. However, \mainAlgo usually outperforms \twotwoapprox when less updates are permitted, suggesting better choices are made when choosing which historical center to preserve.

\spara{Setup 2.} Results are reported in Figure \ref{fig:setup2} for \mainAlgo, $k=30$ and two datasets, with results for \twotwoapprox and additional results for \mainAlgo are reported in Appendix \ref{appendix:extraPlots}. Here, $\mainAlgo\-\alpha\%$ represents the problem instance where $b=\frac{\alpha}{100} n$. The observations made previously in the previous setup carries in this temporal setup: \resilient performs poorly due to each data slice being quite dissimilar with the preceding one, and \carv usually performs better than \mainAlgo but reassigns all the points.

Note that more budget does not guarantee better performance on each individual data slice, as the set of historical centers depends on the output of the algorithm on the previous timestep. However, we notice that more budget allows the algorithm to perform better on most slices: in the Uber dataset, $\mainAlgo\-60\%$ outperforms $\mainAlgo\-20\%$ $80\%$ of the time.

\section{Conclusion and future work}
\label{section:conclusion}

We introduced a novel family of problems \labcons  based on the notion of consistency, and we proposed two constant-factor approximation algorithms in the context of $k$-center clustering.
Our main algorithm achieves a factor $3$ approximation in polynomial time against a lower bound of $2$, leaving open the question of closing this gap. It would also be valuable to study consistency applied to $k$-median or $k$-means, or to different data mining problems.

\section{Acknowledgments}
This research was supported by 
the ERC Advanced Grant REBOUND (834862),
the Swedish Research Council project ExCLUS (2024-05603),
and the Wallenberg AI, Autonomous Systems and Software Program (WASP) funded by the Knut and Alice Wallenberg Foundation.
Some of the computations were enabled by the National Academic Infrastructure for Supercomputing in Sweden (NAISS) and Swedish National Infrastructure for Computing (SNIC) partially funded by the Swedish Research Council through grant agreements no.~2022-06725 and 2018-05973


\section*{Ethics statement}

This work is theoretical in nature and does not involve the collection or use of private, sensitive, or personally identifiable data. No studies involving human subjects were conducted. The methods developed focus on clustering from an algorithmic perspective, and we are not aware of any direct negative societal or ethical implications. While clustering techniques can in principle be applied in sensitive contexts, our contribution is methodological and abstract, and we leave considerations of application-specific impacts to future work.

\section*{Reproducibility statement}

Complete proofs of all theorems presented in the main body of the paper are available in Appendix~\ref{Appendix:ommitedProofs}.
Moreover, all experimental results presented in the main body of the paper and in the appendix were produced using code publicly available at \url{https://github.com/tmaretteKTH/labelConsistentClustering/}.
The archive contains the code to download and preprocess the datasets, the implementation of all methods used in the experiments, as well as necessary code to reproduce the figures. All the plots presented here are, given enough time and compute resources, reproducible using one script.

\section*{LLM usage}

An LLM tool was used solely for light editing tasks such as grammar checking, typo correction, and other minor polishing.
No LLM nor any other kind of generative AI was used for the development of our research ideas, literature review, implementation of our methods, and analysis of results.

\bibliography{arxiv}
\clearpage

\appendix

\section{Pseudocode for \twotwoapprox}
\label{appendix:peusodcode}
Pseudocode for \twotwoapprox can be found in Algorithm~\ref{alg:twotwoapprox}, and pseudocode for \carv in Algorithm~\ref{alg:fft}.

\begin{minipage}[t]{0.57\textwidth}
\vspace{-15pt}
\begin{algorithm}[H]
\caption{\twotwoapprox \hfill\Comment{2-approximation in FPT time}}\label{alg:twotwoapprox}
\begin{algorithmic}[1]
\Require an instance $\cI=\consinst$ of \ourproblem,  optimal radius \optRadius
\Ensure consistent clustering $\AlgCluster=(\AlgClustercenter, \lalg)$ 
\State guess $H^*$ the set of preserved historical center in an optimal solution\label{alg:guesshist} 
\State $X_B \gets X \setminus \cup_{h \in H^*} \ball(h,r^*)$\label{alg:fpt:xb}
\State $S' \gets \carv(X_B, 2\optRadius)$\label{alg:fpt:fft}
\State $\AlgClustercenter\gets H^* \cup S'$
\For{$\point \in X$}\label{alg:fpt:assgn}
    \If{$\lh(x) \in H^*$ and $d(x,\lh(x))\le \optRadius$} 
      \State $\lalg(\point) = \lh(\point)$
    \Else\ $\lalg(\point) \gets \arg\min_{c\in \AlgClustercenter} \dist(\point, c)$\EndIf
\EndFor
\State \textbf{return} $\AlgCluster=(\AlgClustercenter, \lalg)$ \label{alg:fpt:ret}
\end{algorithmic}
\end{algorithm}
\end{minipage}
\hfill
\begin{minipage}[t]{0.33\textwidth}
\begin{algorithm}[H]
\caption{\carv}\label{alg:fft}
\begin{algorithmic}[1]
\Require Set $X'$ of points,  $r \in \mathbb{R}_{\ge 0}$
\Ensure Set $S' \subseteq X'$ 
\State $S' \gets \emptyset$
\While{$X' \neq \emptyset$}
    \State pick arbitrary point $s' \in X'$ 
    \State add $s'$ to $S'$
    \State $X' \gets X' \setminus \ball(s', $$r$)\label{fft:pickpoint}
\EndWhile
\State \Return $S'$
\end{algorithmic}
\end{algorithm}
\end{minipage}

\section{Omitted proofs}
\label{Appendix:ommitedProofs}

\subsection{Proof of~\Cref{thm:hardness}}\label{ss:hardproof}
    The proof follows from a simple polynomial time reduction from \kcenter. Given an instance $\cJ=((X,d),k)$ of \kcenter, we construct an instance $\cI=\consinst$ as follows: we pick an arbitrary set $\histClustercenter \subseteq X$ of size $k$,   assign every point $x \in X$ to a closest center in $\histClustercenter$, i.e., $\lh(x)= \argmin_{h \in \histClustercenter} d(x,h)$, and finally, set $b=|X|$. First, note that this is a polynomial time reduction in $|X|$. Next, we claim that the optimal cost of $\cI$ is equal to the optimal cost of $\cJ$. Towards this, let $r^*_\cI$ and $r^*_\cJ$ be the optimal costs of $\cI$ and $\cJ$, respectively. It is easy to see that $r^*_\cJ \le r^*_\cI$ since any feasible solution to $\cI$ is also a solution to $\cJ$. For the other direction, note that any  solution to $\cJ$ is also a feasible solution to $\cI$ since $b=|X|$. Hence, we have $r^*_\cJ = r^*_\cI$. $i)$  and $ii)$ follows since, for any $\epsilon>0$, \kcenter is both \NP-hard and \wtwo-hard w.r.t. parameter $k$ to approximate to a factor $(2-\epsilon)$~\cite{DBLP:books/daglib/0004338}.
    Finally, for $iii)$,  note that any  solution to $\cJ$ 
    reassigns at most $n \le g(n)=g(b)$, as $b=n$, points, and hence is a feasible solution to $\cI$, as well.\hfill\qed
\subsection{Proof of~\Cref{theorem:2:2}} \label{ss:fptproof}
\twotwoapprox is our algorithm for \ourproblem, which is described in~\Cref{alg:twotwoapprox}. 
Let $\optCluster=(\optClustercenter,\lopt)$ be a fixed but unknown optimal solution to $\cI$ with cost $\optRadius$.
For the analysis, we assume that the algorithm has a correct guess of $\optRadius$ and $H^*=\optClustercenter \cap \histClustercenter$. 
We first claim that $\clusterDist(\AlgCluster,\histCluster) \le b$. Towards this, let $X'_B = \{x \in X \vert \lh(x) \notin H^* \lor d(x,\lh(x))> \optRadius\}$. Then, note that $\optCluster$ has to reassign every $x \in X'_B$ to a center other than $\lh(x)$ in $\optClustercenter$ since $d(x,\optClustercenter) \le \optRadius$. Since, $\clusterDist(\optCluster,\histCluster)\le b$, we have $|X'_B| \le b$. The claim follows since~\Cref{alg:twotwoapprox}  reassigns points only in $X'_B$ in the \textbf{for} loop.

Next, we claim that $|C| = |H^* \cup S| \le k$. Towards this, let $\ell = |H^*| \le k$, then we have to show $|S| \le k-\ell$. Consider the set $X_B$ defined in~\Cref{alg:fpt:xb}, and assume $X_B \neq \emptyset$. Let $\optClustercenter_B = \optClustercenter \setminus H^*$. Then, note that $\optClustercenter_B \neq \emptyset$ since otherwise $\costp{\infty}(\optCluster)>\optRadius$, a contradiction to $\optCluster$. Furthermore, for every $x \in X_B$, it holds that $d(x,\optClustercenter_B) \le \optRadius$. Therefore, $|S| \le |\optClustercenter_B| \le k -\ell$, since $(X_B,2\optRadius)$ picks at most one point from each cluster of $\optClustercenter_B$ (\Cref{fft:pickpoint}).

Finally, for the approximation guarantee, note that for $x \in X \setminus X_B$, we have $d(x,\labc(x)) = d(x,\lopt(x)) \le \optRadius$, while for $x \in X_B$, we have $d(x,\labc(x))=d(x,C) \le d(c,S) \le 2\optRadius$.
For the running time, it is easy to see that the algorithm runs in time $O(nk)$, since $\carv(X_B,2\optRadius)$ runs in time $nk$.

Guessing of $H^*$ can be done by considering every subset of $\histClustercenter$ as a candidate for $H^*$, which results in $2^k$  iterations of~\Cref{alg:twotwoapprox}. The algorithm also needs to verify whether or not the guesses of $H^*$  and $r^*$ are correct, which can be done by checking if the solution $\AlgCluster$ satisfies $|C| \le  k$ and $\clusterDist(\AlgCluster,\histCluster)\le b$. Finally, the algorithm returns  a minimum-cost feasible solution.  Therefore, the overall running time is bounded by $O(n^2 \log n + 2^k nk \log n)$. \hfill\qed

\subsection{Proof of~\Cref{lem:helperlemma}}\label{ss:lmoneproof}
 \begin{figure}
	\begin{center}
		\includegraphics[width=0.55\textwidth]{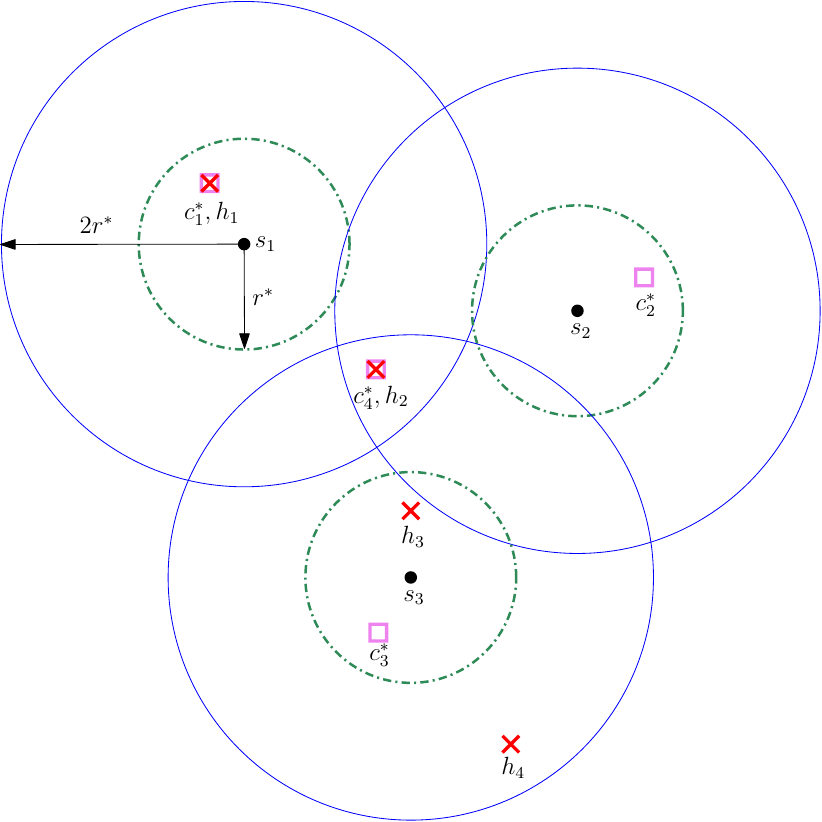}
	\end{center}
        \caption{In the above figure, historical centers are represented by cross, optimal centers are represented by squares, and the elements of $S$ are represented as discs. For simplicity, we do not show remaining points of $X$. It can be seen that $H^*=\{h_1,h_2\}$, while $H^*_c=\{h_1\}$, and $H^*_f=\{h_2\}$. Also, $N_\histClustercenter(s_1)=\{h_1\}$, while $N_\histClustercenter(s_3)=\{h_3\}$. Furthermore, $\Gamma^*_{s_1} = \{h_1\} =H^*_c$, while all other $\Gamma^*$s are empty. Finally, suppose $w(h_4) > w(h_2)$, then the center set chosen by~\Cref{alg:3approxnew} is $C=\{h_1,s_2,h_3,h_4\}$, and hence $w(C \cap \histClustercenter) = w(h_1)+ w(h_3) + w(h_4) > w(h_1) + w(h_2) = w(\optClustercenter \cap \histClustercenter)$
        , and $\costp{\infty}(\AlgCluster) \le 3\optRadius$.}
        \label{fig:f1}
\end{figure}
Consider $s \in S$. Then, we know that $\Gamma^*_s$ are precisely the historical centers present in the optimal solution that are covered by $s$ (see~\Cref{fig:f1} for an illustration). On the other hand, let $\Gamma_s$ be the historical centers picked by the algorithm in $C'_0 \cup C'_1$, i.e., $\Gamma_s = \hat{\Gamma}_s \cup \{\hat{s}\}$.
        Note that, for $s\neq s' \in S$, we have that $N_\histClustercenter(s) \cap N_\histClustercenter(s') = \emptyset$, since $d(s,s')>2r^*$, and hence, $\hat{s} \neq \hat{s}'$ and $\Gamma^*_s \cap \Gamma^*_{s'} = \emptyset$,
 which in turn implies that ${\Gamma}_s \cap {\Gamma}_{s'} = \emptyset$. Therefore, we have that, the total number of historical centers in $C'_0 \cup C'_1$ is lower bounded by $\sum_{s \in S_\gamma} |\Gamma_s| = \sum_{s \in S_\gamma} |\Gamma^*_s| =  |H^*_c|$. 
Recall that, $H^*$ is the set of historical centers present in the optimal solution $\optClustercenter$, and $w(H^*_c) + w(H^*_f)=w(H^*) \ge n-b$. 
       We start by showing that the weight of the historical centers present in $C'_0 \cup C'_1$ is at least the weight of the hidden historical centers, $w(H^*_c)$.
        \begin{claim}
        \label{cl:help:reasgn}
         $w(\histClustercenter \cap (C'_0 \cup C'_1)) \geq w(H^*_c)$.
        \end{claim}
        \begin{proof}
            First note that $w(\histClustercenter \cap (C'_0 \cup C'_1)) \ge w(\cup_{s \in S_\gamma} \Gamma_s)$ since $\histClustercenter \cap (C'_0 \cup C'_1))= \cup_{s \in S} \Gamma_s \supseteq  \cup_{s \in S_\gamma} \Gamma_s$. Hence, it is sufficient to show that $w(\cup_{s \in S_\gamma}\Gamma_s) \ge w(\cup_{s \in S_\gamma}\Gamma_s^*)=w(H^*_c)$. 
            Towards this, we show that for $s\in S_\gamma$, it holds that $w(\Gamma_s) \ge w(\Gamma^*_s)$, which implies that $w(\cup_{s \in S_\gamma}\Gamma_s) = \sum_{s \in S_\gamma}w(\Gamma_s) \ge  \sum_{s \in S_\gamma}w(\Gamma^*_s)=w(\cup_{s \in S_\gamma}\Gamma^*_s) =w(H^*_c)$, as required.
            
            Consider the \textbf{for} loop execution at~\Cref{alg:helperalgo:for} for $s \in S_\gamma$, and 
            consider $\Gamma_s= \{\hat{s}\} \cup \hat{\Gamma}_s$ and let $\Gamma^*_s = (h^1_s, h^2_s, \dots h^{|\Gamma^*_s|}_s)$ be ordered in non-increasing order of the weights of the historical optimal centers.

            Then, note that $w(\hat{s}) \ge w(h^1_s)$ as the set $\Gamma^*_s$ is available for the algorithm to be picked for $\hat{s}$. Now consider the case when $\hat{s} \notin \Gamma^*_s$. 
            Since all the centers in $\Gamma^*_s$ are available to be picked by the algorithm in~\Cref{alg:helperalgo:pickadd}, we have that $w(\hat{\Gamma}_s) \ge w(\{h_s^1,\dots,h_s^{|\Gamma^*_s|-1}\})$, and thus, $w(\Gamma_s) = w(\hat{s}) + w(\hat{\Gamma}_s) \ge w(h^1_s) + w(\{h_s^1,\dots,h_s^{|\Gamma^*_s|-1}\}) \geq w(\Gamma^*_s)$, as required. Now consider the case when $\hat{s} \in \Gamma^*_s$, and therefore $\hat{s}=h^1_s$ since it is the highest weight historical center in $\Gamma^*_s$. This means that, the centers $\{h^2_s,\dots, h^{|\Gamma^*_s|}_s\}$ are available for the algorithm to pick in~\Cref{alg:helperalgo:pickadd}. Therefore, $w(\Gamma_s) \ge w(h^1_s) + w(\{h_s^2,\dots,h_s^{|\Gamma^*_s|}\} = w(\Gamma^*_s)$, as required.
            \end{proof}

Next, we show that $|C'| \le k$.
        Towards this, suppose $|S|=\ell$, and consider the optimal clusters $\{\pi^*_i\}_{i \in [k]}$. We say that a optimal cluster $\pi^*_i$ is \emph{hit} by $S$ if $\exists s \in S$ such that $s$ belongs to cluster $\pi^*_i$. Since, the size of $S$ is $\ell$, and elements of $S$ belong to different clusters of the optimal solution, we have that $S$ hits exactly $\ell$ optimal clusters. Therefore, the number of \emph{unhit} optimal clusters is $k-\ell$. On the other hand, $|C'_1 \cup C'_2|=|C'_1| + |C'_2| = \sum_{s \in S_\gamma} (|\Gamma^*_s|-1) + |H^*_f|$ is upper bounded by the number of unhit optimal clusters, while $|C'_0| = |S| =\ell$. Hence, we have $|C'| = |C'_0| + |C'_1| + |C'_2| \le k$, as desired.\hfill\qed

        Now, we show that $w(\histClustercenter \cap C') \ge n-b$. 
        Consider the hidden historical optimal centers, $H^*_f$, to $S$. Since, it holds that $H^*_f \cap (\cup_{s \in S} N_\histClustercenter(s))=\emptyset$, we have that $ H^*_f \cap (C'_0 \cup C'_1) =H^*_f \cap (\cup_{s \in S} \Gamma_s) = \emptyset$. Thus,  all the centers in $H^*_f$ are available to be picked by the algorithm in~\Cref{alg:helperalgo:pickhid} for $C'_2$. This means $w(C'_2) \ge w(H^*_f)$, and hence 
            \begin{align*}
            w(\histClustercenter \cap C') &= w(\histClustercenter \cap (C'_0 \cup C'_1)) +w(C'_2)\\
            &\ge w(H^*_c)+w(H^*_f)\\
            &= w(H^*)\\
            &\ge n-b,   
        \end{align*}
        where the first equality follows since the sets $C'_0,C'_1,C'_2$ are pairwise disjoint and the fact that $\histClustercenter \cap C'_2 = C'_2$, and the first inequality follows due to the above claim.
        Now, consider the \textbf{for} loop in~\Cref{alg:helper:assgn} that assigns points to a center in $C$. Note that for every $x \in X$, it assigns $x$ to $\lh(x)$ if $\lh(x) \in C'$ and $d(x,\lh(x)) \le \optRadius$. Therefore, for every historical center $h \in C'\cap \histCluster$, the number of points assigned to $h$ is at  least $w(h)$. Hence, the number of points reassigned by~\Cref{alg:helperalgo} is at least $w(\histClustercenter \cap C')\geq n-b$, finishing the proof of the lemma.
        \hfill\qed
\subsection{Proof of~\Cref{lem:3apxvshelper}}\label{ss:lemtwoproof}
 
    Consider the center sets $\AlgClustercenter=C_0 \cup C_1$ and $\AlgClustercentera=C'_0 \cup C'_1\cup C'_2$ returned by~\Cref{alg:3approxnew} and~\Cref{alg:helperalgo}, respectively. Then, note that $C_0 = C'_0$. Consider~\Cref{algo:3approx:c1} of~\Cref{alg:3approxnew} that constructs $C_1$. Since $C_0=C'_0$, we have that all the historical centers in the set $C'_1 \cup C'_2$ picked by~\Cref{alg:helperalgo} are available for~\Cref{alg:3approxnew} as candidate centers for $C_1$ in~\Cref{algo:3approx:c1}.
    Furthermore, $|C'_1\cup C'_2| \leq (k-|S|) = (k-|C_0|)$. Therefore, we have that $w(C_1) \ge w(\histClustercenter \cap (C'_1\cup C'_2))$, and hence
    \begin{align*}
        w(\histClustercenter \cap C) &= w(\histClustercenter \cap (C_0 \cup C_1))\\
        &= w(\histCluster \cap C_0) + w(C_1)\\
        &\geq  w(\histCluster \cap C'_0) + w(\histClustercenter \cap (C'_1\cup C'_2))\\
        &\geq w(\histCluster \cap C')\\
        &\geq n-b.
    \end{align*}
    The lemma follows since in the assignment routine (\textbf{for} loop in~\Cref{algo:3approx:assgn}) of~\Cref{alg:3approxnew},  the number of points assigned to $h\in C\cap \histCluster$ is at  least $w(h)$.\hfill\qed

\subsection{Proof of \Cref{thm:mainthm}}\label{ss:mainthmproof}
\Cref{alg:3approxnew} gives a $3$-approximation for \ourproblem\ by \Cref{thm:3approx}, but it requires the optimal radius $r^*$ as an additional input.  
A straightforward approach is to try all $n^2$ pairwise distances as guesses for $r^*$ and return the minimum-cost feasible solution.  
However, we can accelerate this by performing a binary search over these $n^2$ candidate values, yielding a total running time of 
$O(n^2 \log n + nk \log n)$, as desired.\qed

\section{Additional experimental results}
\subsection{Datasets description}
\label{appendix:datasets}
The \emph{Abalone} dataset \cite{abalone_1} dataset records different measurements on abalones. 
We keep all features, except the sex 
and the number of rings.
The dataset contains 7 features and 4176 entries.

The \emph{Electric Consumption} dataset \cite{individual_household_electric_power_consumption_235}
dataset contains 207.5 million  power measurement in Sceaux, France, 
between Dec.~2006 and Nov.~2010. 
We retain measurements occurred during the first 20 months, for a total of 439820 entries, and use 7 numerical features. 

The \emph{OnlineRetail} dataset \cite{online_retail_352} records 54.2 million transactions of 
online retail, occurring from Dec.~2010 to Dec.~2011. We retain the two numerical features, and transactions occurring between Jan. 1, 2011 and Jan. 20, 2011, for a total of 28300 transactions.

The \emph{Twitter} dataset \cite{twitter_geospatial_data_1050} record seven days of geo-tagged Tweet data from the United States, sent between Jan 12, 2013 and Jan 18, 2013. We keep tweets sent during the first 20 hours, for a total of 289860 tweets.

The \emph{Uber} dataset\footnote{ \url{https://www.kaggle.com/datasets/fivethirtyeight/uber-pickups-in-new-york-city}} contains around 18.8 million Uber pickups in New York City 
from April to June 2015. We retain pickups between Jun 1, 2014 and Jun 20, 2014, for a total of 290200 pickups.

In both the Twitter dataset and the Uber dataset, we convert the angular Geo location to Cartesian coordinates, for a total of 3 features. 

\subsection{Experiments on $\varepsilon$-close point sets}
\label{appendix:eps_close_instances}
This additional experiments was conducted to study the empirical performance of our algorithm on instances where when data are affected by some noisy perturbation, which are instances where \resilient is designed to perform well.

We use the Uber dataset, and follow the same pre-processing method as by \citet{ahmadian2024resilient}: 
we consider the pickup locations of the first two days of June 2014, and convert the angular Geo location to Cartesian coordinates. Then, we find a minimum-weight perfect matching between points of the first and second day, and remove unmatched locations as well as matched locations located more than 1\,km apart. Finally, we denote by $X'$ the dataset for the first day, and $X$ the dataset for the second day. We run \textsc{H{\large ist}} on $X'$ to obtain \histCluster, and convert the historical centers \histClustercenter  to their matched points in $X$. Then, we run \textsc{A{\large lg}} on $X$ with \histCluster as historical clustering.

\begin{figure*}
\centering
\begin{subfigure}{\textwidth}
\includegraphics[width=.33\textwidth]{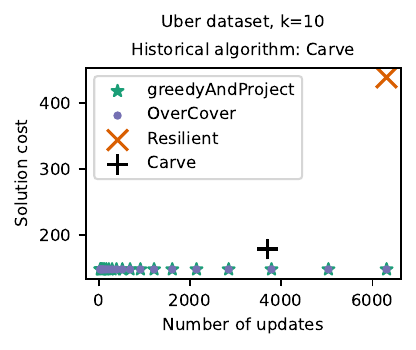}\hfill
\includegraphics[width=.33\textwidth]{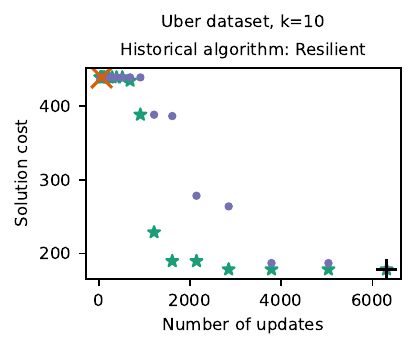}\hfill
\includegraphics[width=.33\textwidth]{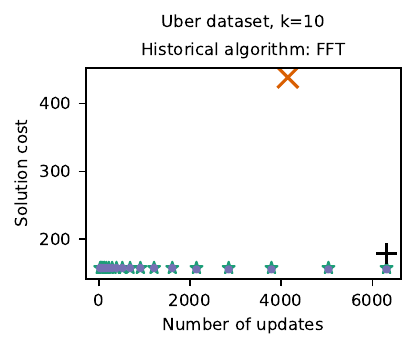}\hfill
\end{subfigure}\hfill
\begin{subfigure}{\textwidth}
\includegraphics[width=.33\textwidth]{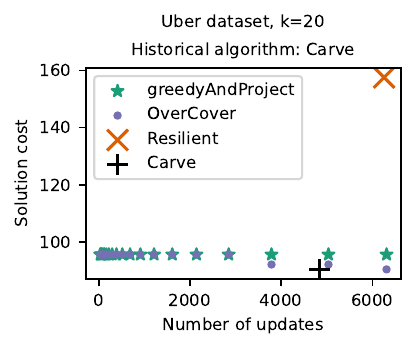}\hfill
\includegraphics[width=.33\textwidth]{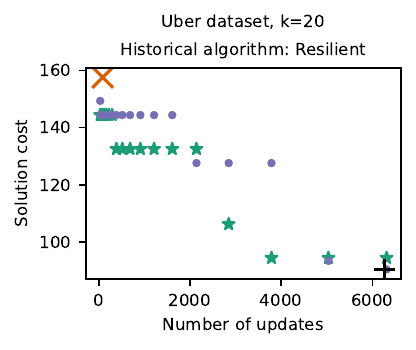}\hfill
\includegraphics[width=.33\textwidth]{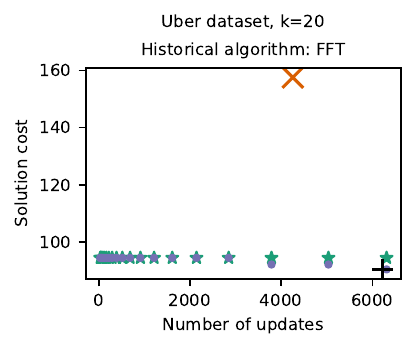}\hfill

\end{subfigure}
\hfill
\begin{subfigure}{\textwidth}
\includegraphics[width=.33\textwidth]{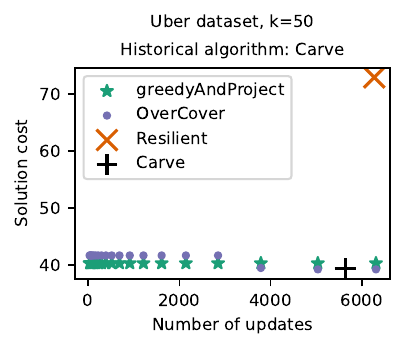}\hfill
\includegraphics[width=.33\textwidth]{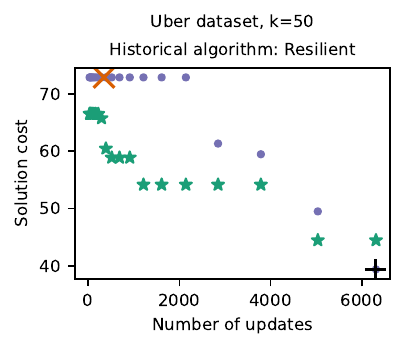}\hfill
\includegraphics[width=.33\textwidth]{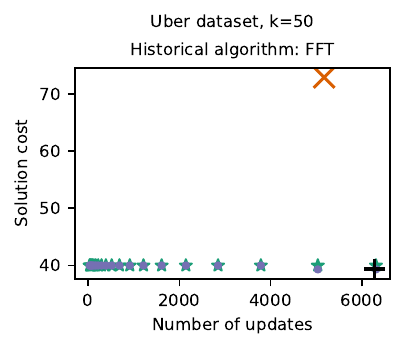}\hfill

\end{subfigure}

\caption{Comparison of our algorithm with the baselines for the third experimental setup.}
\label{fig:additional}
\end{figure*}

\mpara{Empirical results.} Results for this setup are reported in Figure \ref{fig:additional}. Since the points in the first and second day datasets are similar, \resilient manages to find similar clusters, provided that \resilient or \gonz is used as historical clustering. For example, \resilient achieves as little as $2\%$ reassigned points when using \resilient historical centers with $k=30$, and $80\%$ when using \gonz as historical clustering. This cross-compatibility exists because \resilient uses \gonz as a subroutine to find some centers.  Note that the consistency is lost when the historical clustering comes from a different clustering algorithm, such as \carv. Note that both \gonz and \carv achieve lower clustering cost than \resilient.

Meanwhile, with both \twotwoapprox and \mainAlgo, we can perform similarly to \resilient and \carv, on both the number of reassigned points and clustering score, with an additional control of in-between number of updates. 
\subsection{Running Time Analysis}
\label{appendix: scalability}
We compare running time of \mainAlgo and \twotwoapprox for different $k$ values, on the Twitter dataset, and report the running times on Figure \ref{fig:scalability}. First, we can notice that the running time seems independent of the number of updates $b$. Furthermore, for fixed value of $k$, the running time scales of both algorithm scales well with respect to the size of the dataset.

\begin{figure*}
\centering
\begin{subfigure}{\textwidth}
\includegraphics[width=.33\textwidth]{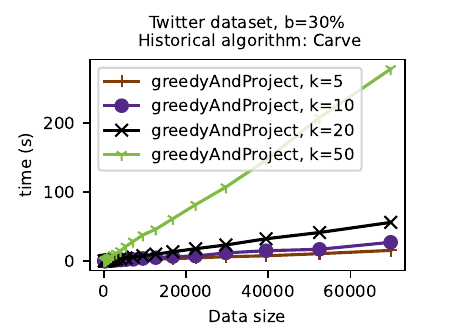}\hfill
\includegraphics[width=.33\textwidth]{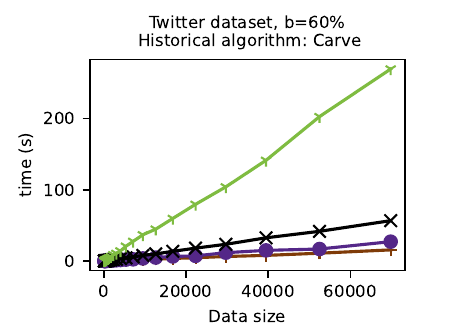}\hfill
\includegraphics[width=.33\textwidth]{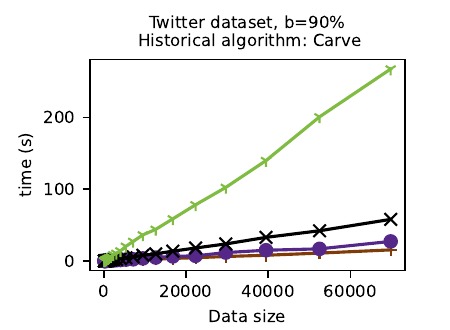}
\end{subfigure}
\hfill
\begin{subfigure}{\textwidth}
\includegraphics[width=.33\textwidth]{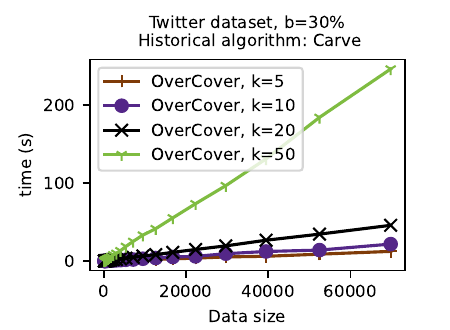}\hfill
\includegraphics[width=.33\textwidth]{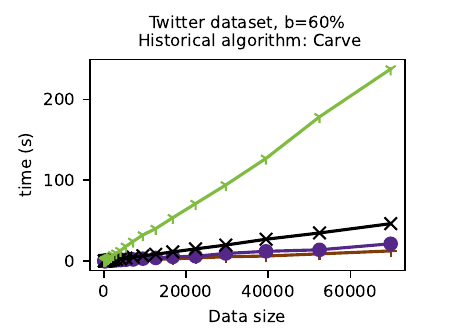}\hfill
\includegraphics[width=.33\textwidth]{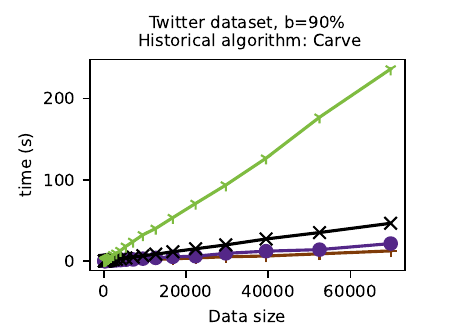}
\end{subfigure}
\caption{Scalability of our methods}
\label{fig:scalability}
\end{figure*}

\subsection{Additional results for experimental setup 1 and 2}
\label{appendix:extraPlots}
We provide additional plots for the experimental setups described in Section \ref{sec:setups}. For setup 1, this includes all combinations of historical clustering and datasets, for $k=10$ in Figure \ref{fig:extraPlots1.1}, $k=20$ in Figure \ref{fig:extraPlots1.2} and $k=50$ in Figure \ref{fig:extraPlots1.3}. For setup 2, this includes plots for the Electricity and Twitter dataset for \mainAlgo in Figure \ref{fig:extraPlots2.1}, and plots for \twotwoapprox in Figure \ref{fig:extraPlots2.2}.

\begin{figure*}
\centering
\begin{subfigure}{\textwidth}
\includegraphics[width=.3\textwidth]{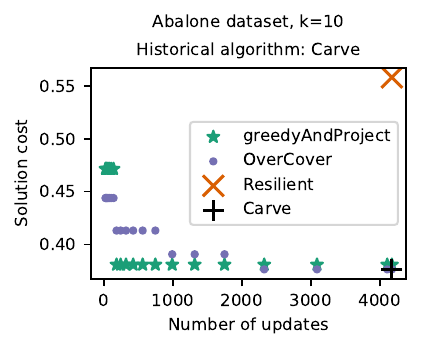}\hfill
\includegraphics[width=.3\textwidth]{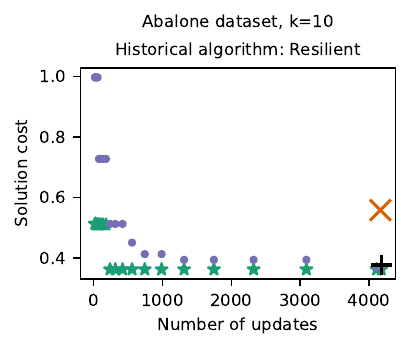}\hfill
\includegraphics[width=.3\textwidth]{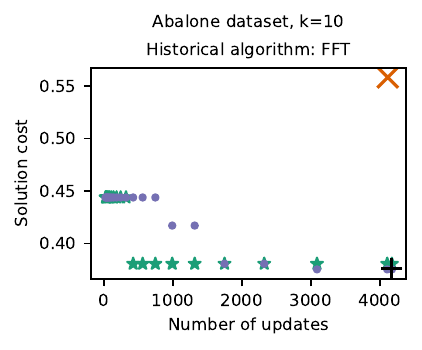}
\end{subfigure}
\hfill
\begin{subfigure}{\textwidth}
\includegraphics[width=.3\textwidth]{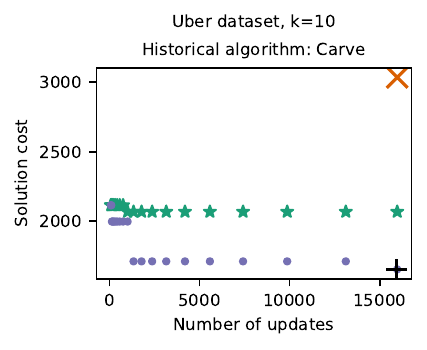}\hfill
\includegraphics[width=.3\textwidth]{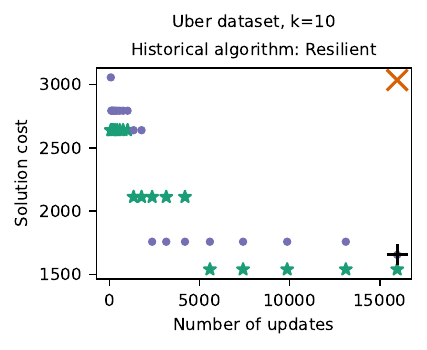}\hfill
\includegraphics[width=.3\textwidth]{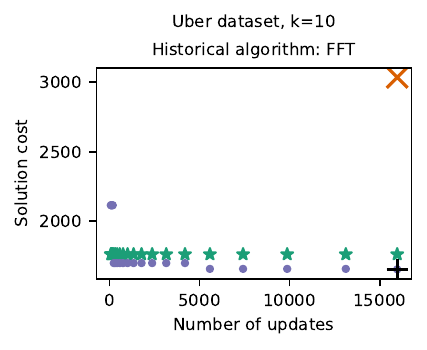}
\end{subfigure}
\hfill
\begin{subfigure}{\textwidth}
\includegraphics[width=.3\textwidth]{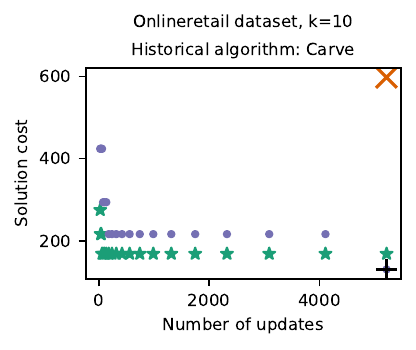}\hfill
\includegraphics[width=.3\textwidth]{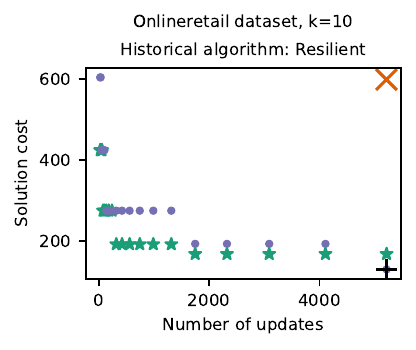}\hfill
\includegraphics[width=.3\textwidth]{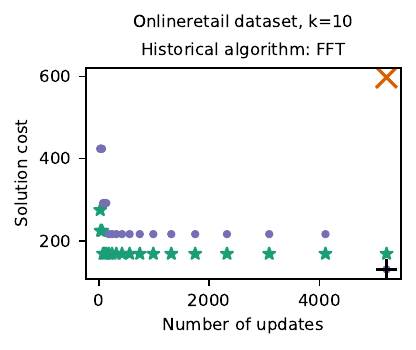}
\end{subfigure}
\hfill
\begin{subfigure}{\textwidth}
\includegraphics[width=.3\textwidth]{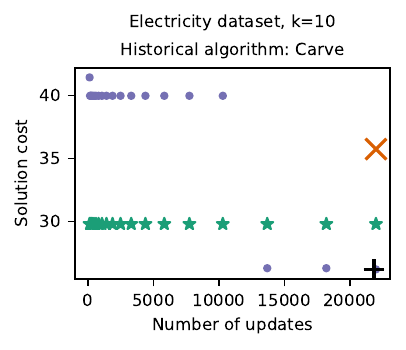}\hfill
\includegraphics[width=.3\textwidth]{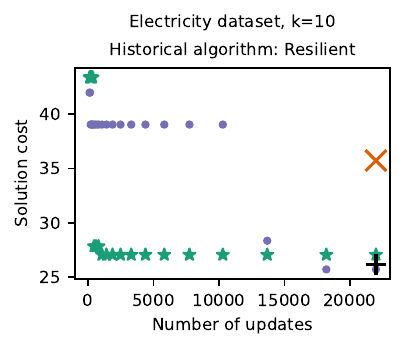}\hfill
\includegraphics[width=.3\textwidth]{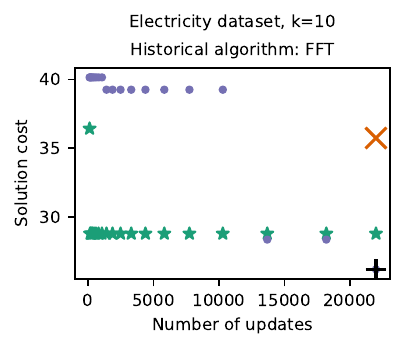}
\end{subfigure}
\hfill
\begin{subfigure}{\textwidth}
\includegraphics[width=.3\textwidth]{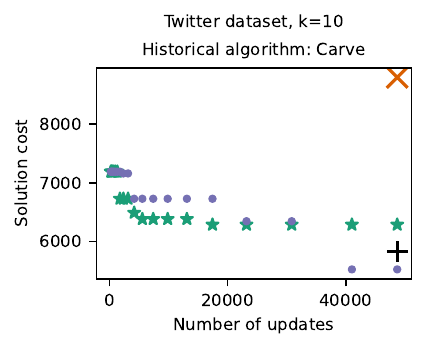}\hfill
\includegraphics[width=.3\textwidth]{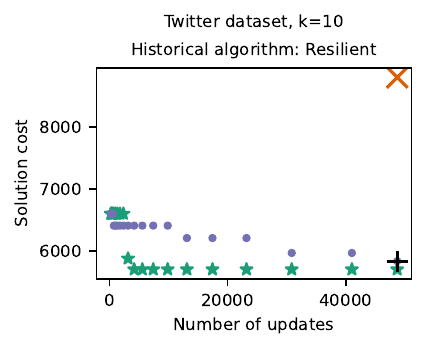}\hfill
\includegraphics[width=.3\textwidth]{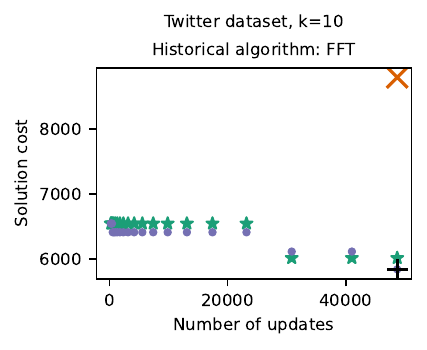}
\end{subfigure}
\caption{Extra plots for the first experimental setup, $k=10$.}
\label{fig:extraPlots1.1}
\end{figure*}

\begin{figure*}
\centering
\begin{subfigure}{\textwidth}
\includegraphics[width=.3\textwidth]{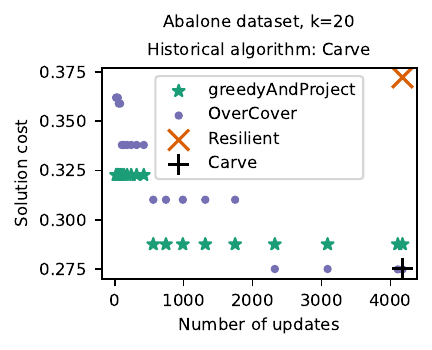}\hfill
\includegraphics[width=.3\textwidth]{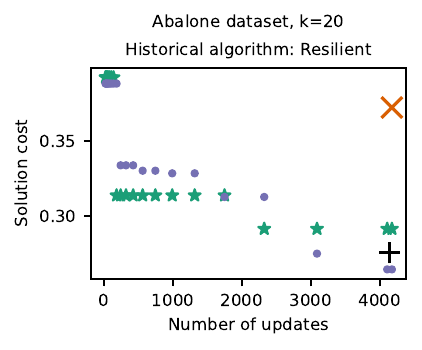}\hfill
\includegraphics[width=.3\textwidth]{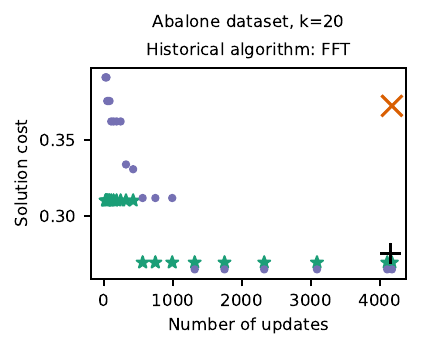}
\end{subfigure}
\hfill
\begin{subfigure}{\textwidth}
\includegraphics[width=.3\textwidth]{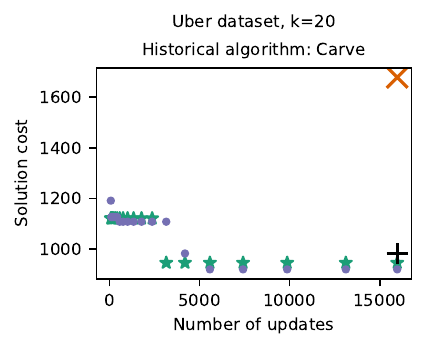}\hfill
\includegraphics[width=.3\textwidth]{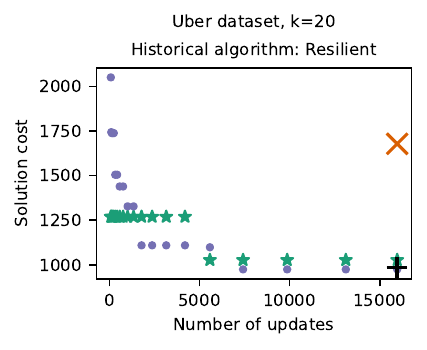}\hfill
\includegraphics[width=.3\textwidth]{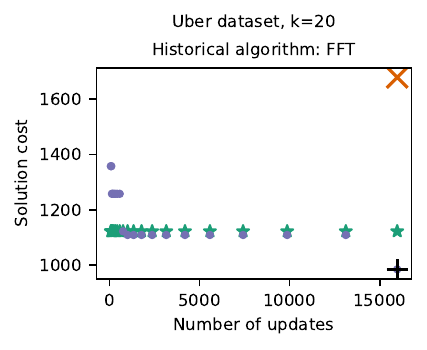}
\end{subfigure}
\hfill
\begin{subfigure}{\textwidth}
\includegraphics[width=.3\textwidth]{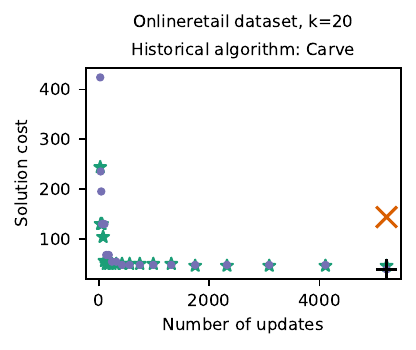}\hfill
\includegraphics[width=.3\textwidth]{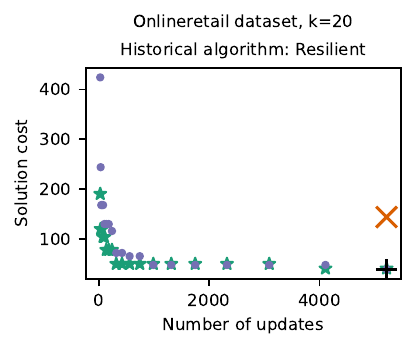}\hfill
\includegraphics[width=.3\textwidth]{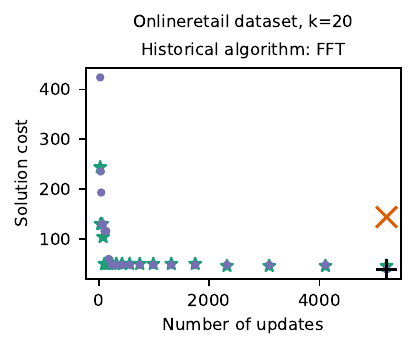}
\end{subfigure}
\hfill
\begin{subfigure}{\textwidth}
\includegraphics[width=.3\textwidth]{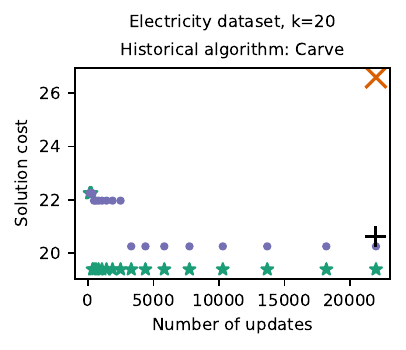}\hfill
\includegraphics[width=.3\textwidth]{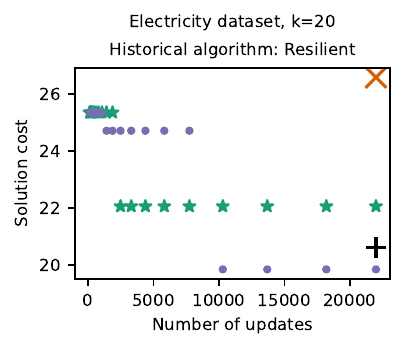}\hfill
\includegraphics[width=.3\textwidth]{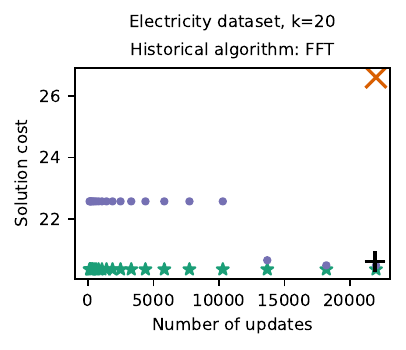}
\end{subfigure}
\hfill
\begin{subfigure}{\textwidth}
\includegraphics[width=.3\textwidth]{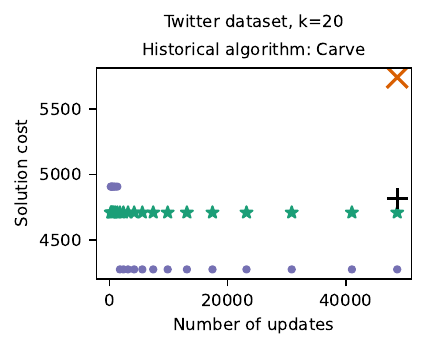}\hfill
\includegraphics[width=.3\textwidth]{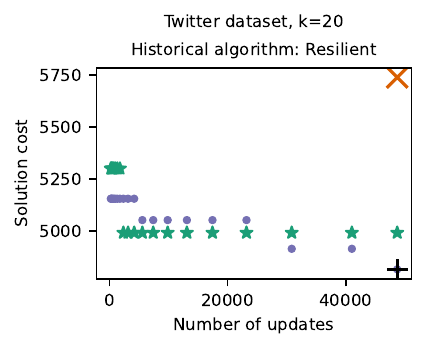}\hfill
\includegraphics[width=.3\textwidth]{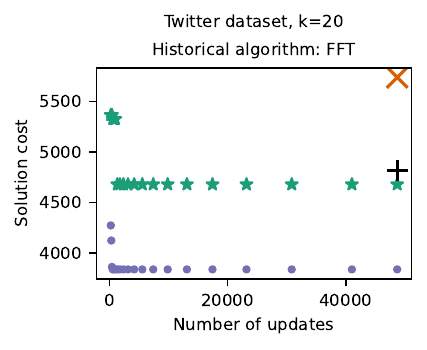}
\end{subfigure}
\caption{Extra plots for the first experimental setup, $k=20$.}
\label{fig:extraPlots1.2}
\end{figure*}

\begin{figure*}
\centering
\begin{subfigure}{\textwidth}
\includegraphics[width=.3\textwidth]{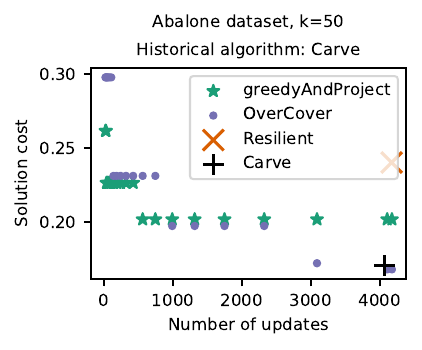}\hfill
\includegraphics[width=.3\textwidth]{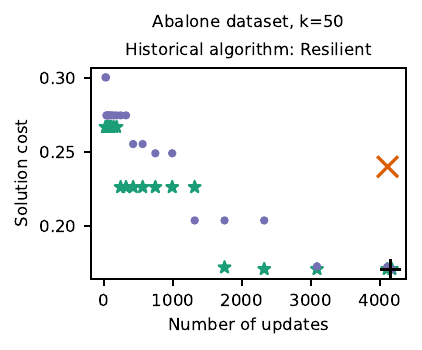}\hfill
\includegraphics[width=.3\textwidth]{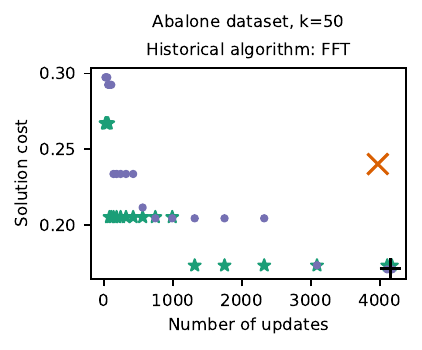}
\end{subfigure}
\hfill
\begin{subfigure}{\textwidth}
\includegraphics[width=.3\textwidth]{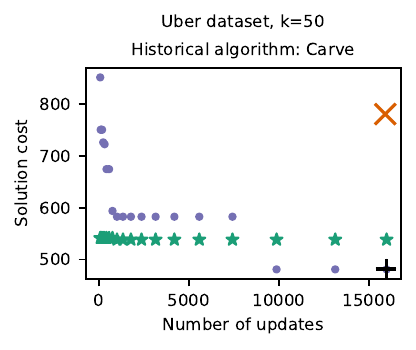}\hfill
\includegraphics[width=.3\textwidth]{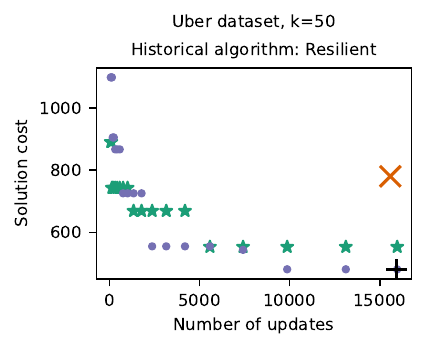}\hfill
\includegraphics[width=.3\textwidth]{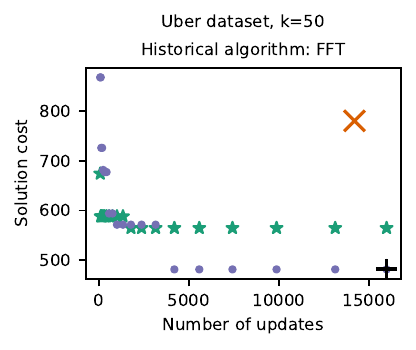}
\end{subfigure}
\hfill
\begin{subfigure}{\textwidth}
\includegraphics[width=.3\textwidth]{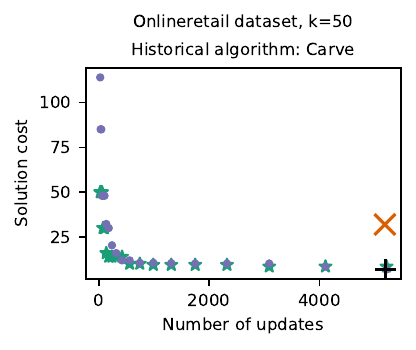}\hfill
\includegraphics[width=.3\textwidth]{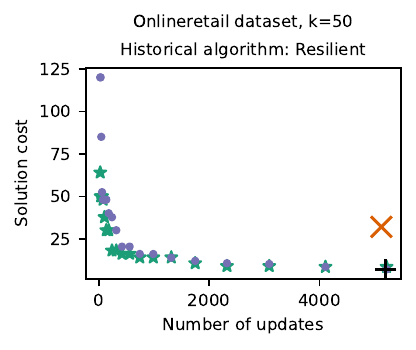}\hfill
\includegraphics[width=.3\textwidth]{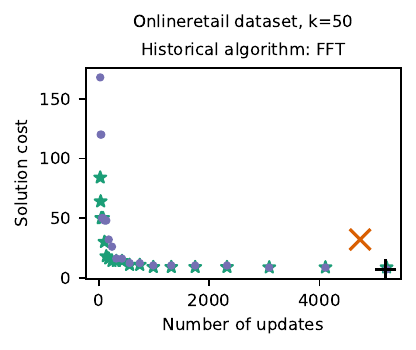}
\end{subfigure}
\hfill
\begin{subfigure}{\textwidth}
\includegraphics[width=.3\textwidth]{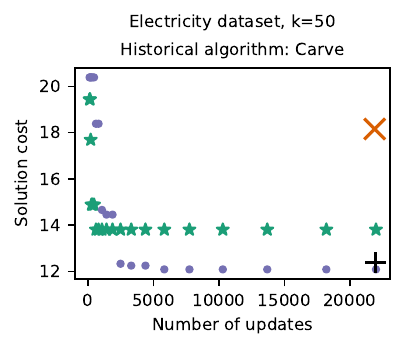}\hfill
\includegraphics[width=.3\textwidth]{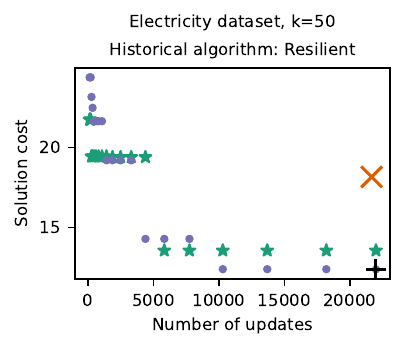}\hfill
\includegraphics[width=.3\textwidth]{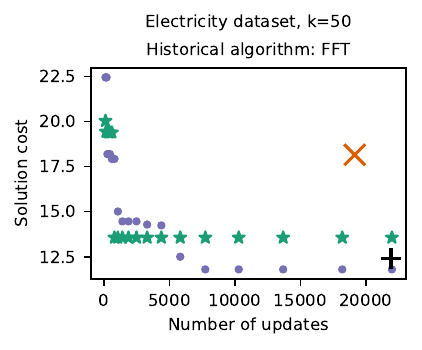}
\end{subfigure}
\hfill
\begin{subfigure}{\textwidth}
\includegraphics[width=.3\textwidth]{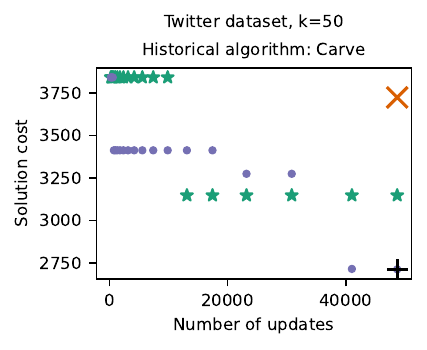}\hfill
\includegraphics[width=.3\textwidth]{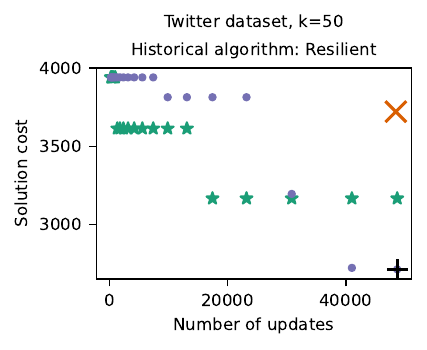}\hfill
\includegraphics[width=.3\textwidth]{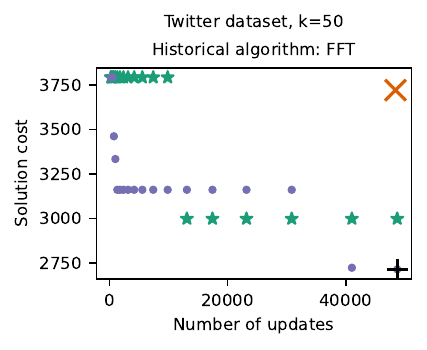}
\end{subfigure}
\caption{Extra plots for the first experimental setup, $k=50$.}
\label{fig:extraPlots1.3}
\end{figure*}

\begin{figure*}
\centering
\begin{subfigure}{\textwidth}
\includegraphics[width=.58\textwidth]{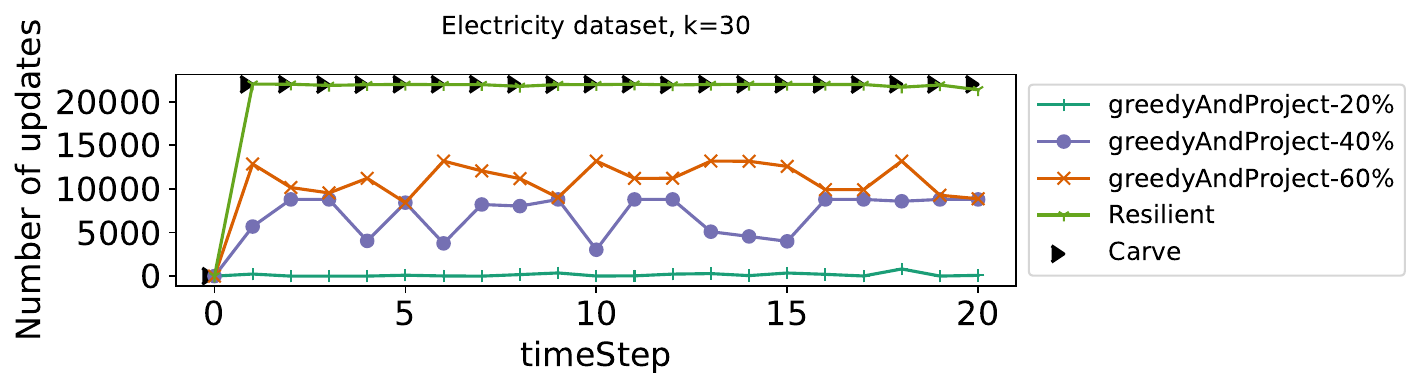}\hfill
\includegraphics[width=.41\textwidth]{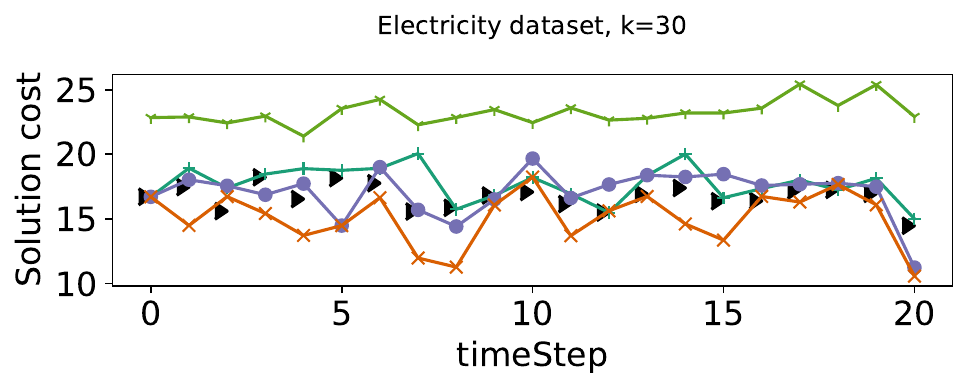}\hfill
\end{subfigure}

\begin{subfigure}{\textwidth}
\includegraphics[width=.58\textwidth]{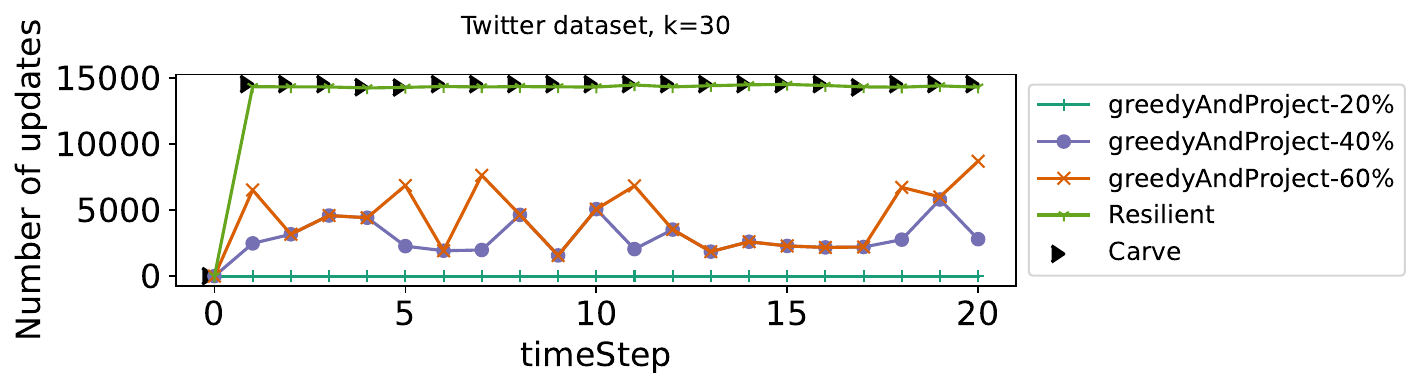}\hfill
\includegraphics[width=.41\textwidth]{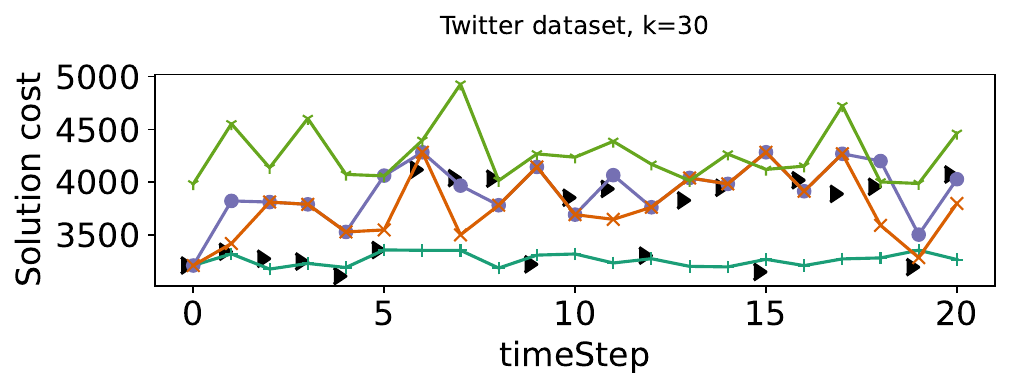}\hfill
\end{subfigure}

\caption{Extra plots for the second experimental setup, for \mainAlgo.}
\label{fig:extraPlots2.1}
\end{figure*}

\begin{figure*}
\centering
\begin{subfigure}{\textwidth}
\includegraphics[width=.58\textwidth]{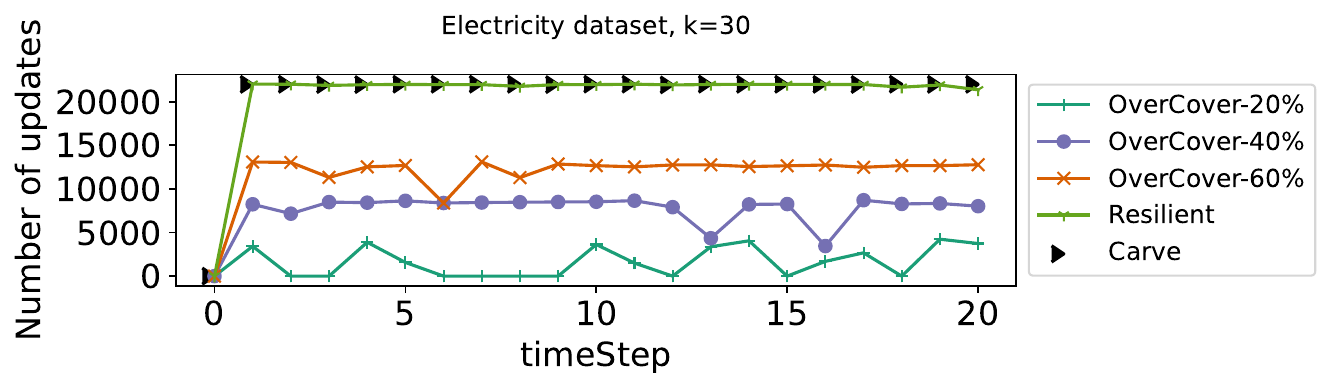}\hfill
\includegraphics[width=.41\textwidth]{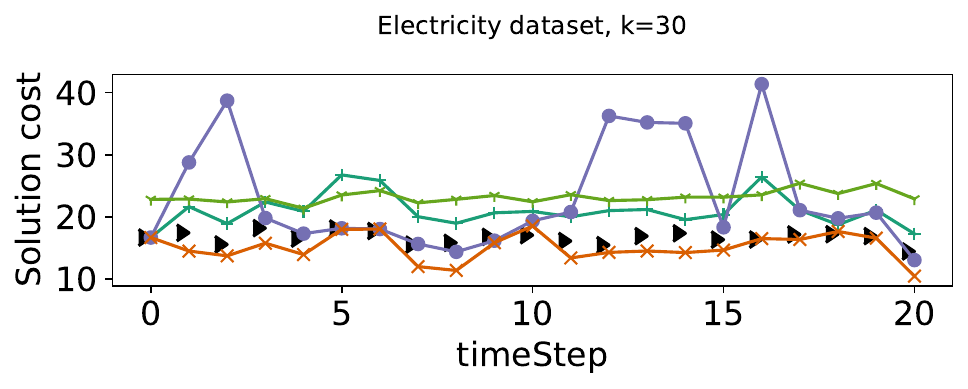}\hfill
\end{subfigure}
\hfill
\begin{subfigure}{\textwidth}
\includegraphics[width=.58\textwidth]{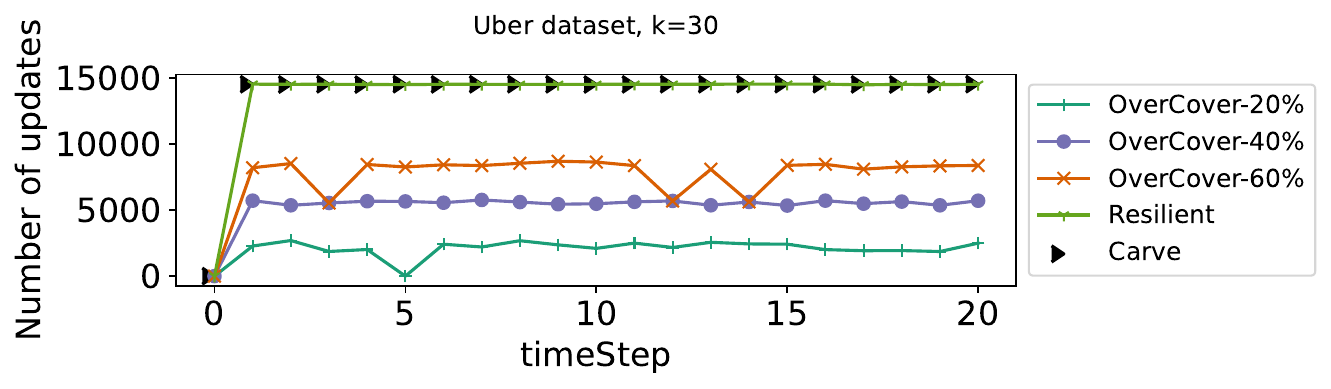}\hfill
\includegraphics[width=.41\textwidth]{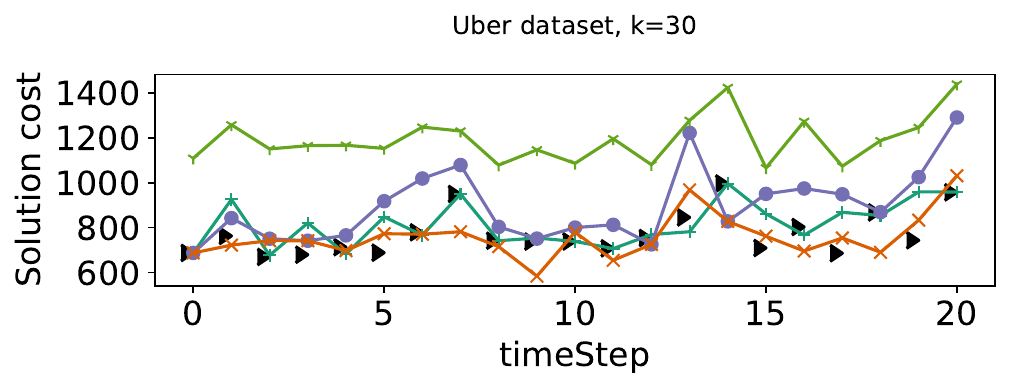}\hfill
\end{subfigure}
\begin{subfigure}{\textwidth}
\includegraphics[width=.58\textwidth]{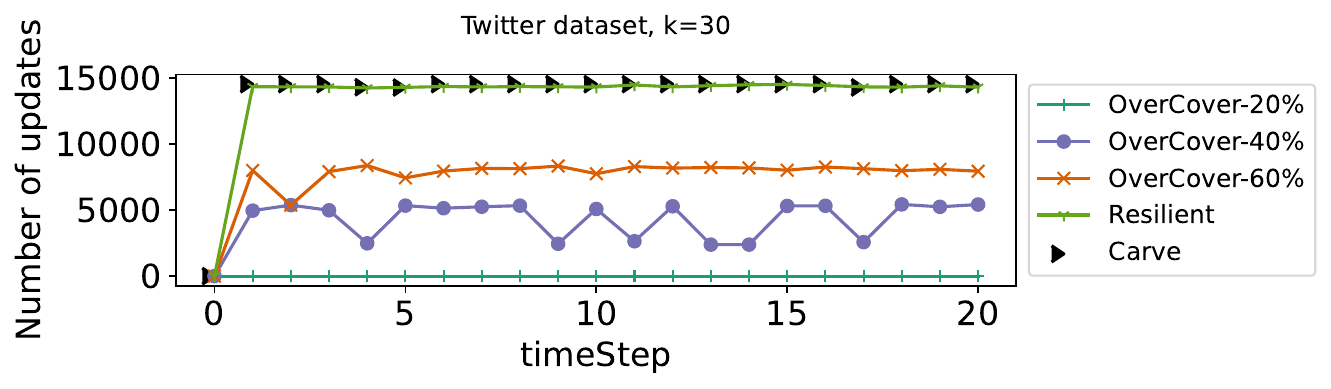}\hfill
\includegraphics[width=.41\textwidth]{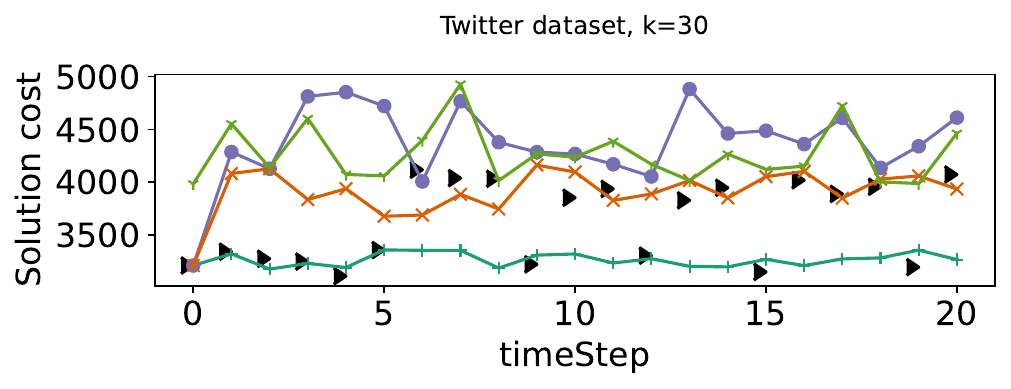}\hfill
\end{subfigure}
\begin{subfigure}{\textwidth}
\includegraphics[width=.58\textwidth]{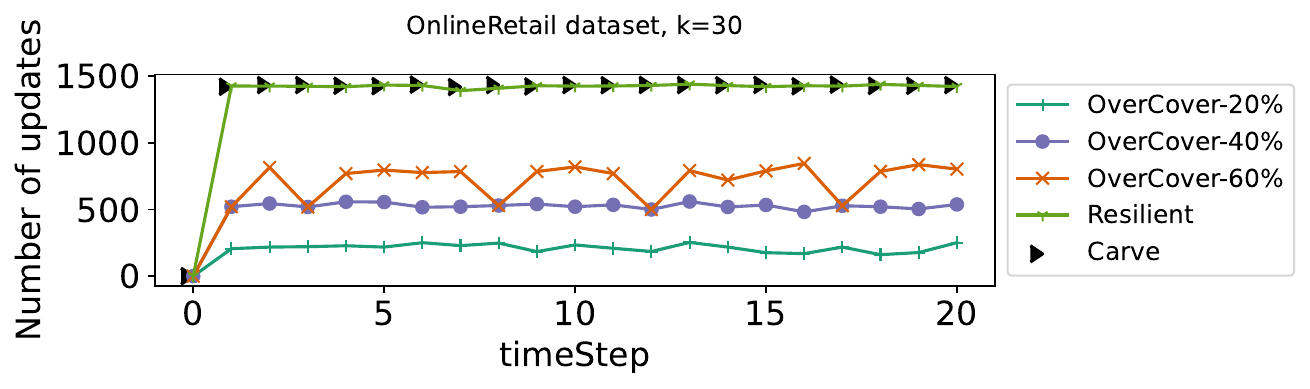}\hfill
\includegraphics[width=.41\textwidth]{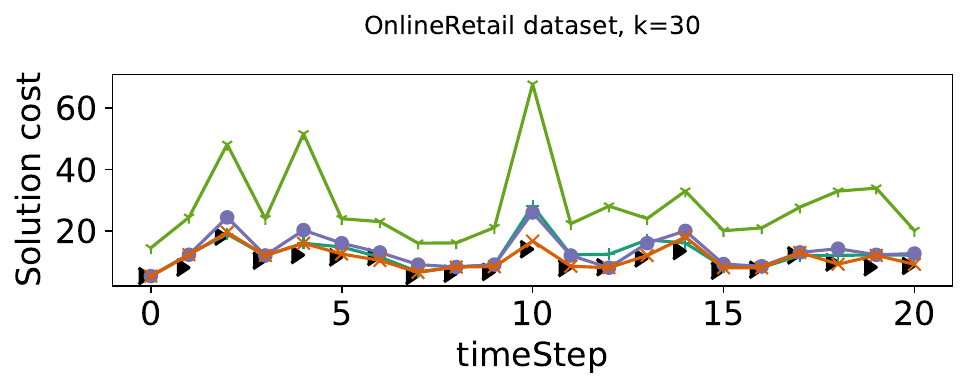}\hfill
\end{subfigure}
\caption{Extra plots for the second experimental setup, for \twotwoapprox.}
\label{fig:extraPlots2.2}
\end{figure*}

\end{document}